\RequirePackage[l2tabu,orthodox]{nag}
\documentclass
[11pt,letterpaper]
{article} 

\usepackage[notes=false,later=false,camera=true]{dtrt}
\usepackage[utf8]{inputenc}
\usepackage{etex}
\usepackage{ stmaryrd }
\usepackage{xspace,enumerate}
\usepackage[T1]{fontenc}
\usepackage[full]{textcomp}
\usepackage[american]{babel}
\usepackage{mathtools}

\usepackage{amsthm}
\usepackage{empheq}

\usepackage[pagebackref]{hyperref}

\hypersetup{
colorlinks=true,
urlcolor=Cerulean,
linkcolor=RoyalBlue,
citecolor=OliveGreen,
linktocpage=true,
}

\usepackage[capitalise,nameinlink]{cleveref}
\crefname{lemma}{Lemma}{Lemmas}
\crefname{fact}{Fact}{Facts}
\newcommand{\colorconstraints}{\text{Color Constraints}}
\crefname{colorconstraints}{(color constraints)}{Color Constraints}
\crefformat{colorconstraints}{#2\colorconstraints#3}
\crefname{indsetconstraints}{(indset constraints)}{IndSet Constraints}
\crefformat{indsetconstraints}{#2$\mathsf{IndSet\ Axioms}$#3}
\crefname{theorem}{Theorem}{Theorems}
\crefname{mtheorem}{Theorem}{Theorems}
\crefname{itheorem}{Theorem}{Theorems}
\crefname{corollary}{Corollary}{Corollaries}
\crefname{claim}{Claim}{Claims}
\crefname{example}{Example}{Examples}
\crefname{algorithm}{Algorithm}{Algorithms}
\crefname{problem}{Problem}{Problems}
\crefname{definition}{Definition}{Definitions}
\crefname{equation}{Eq.}{Eq.}
\crefname{strategy}{Strategy}{Strategies}

\usepackage{paralist}
\usepackage{turnstile}
\usepackage{mdframed}
\usepackage{tikz}
\usepackage{caption}
\DeclareCaptionType{Algorithm}
\usepackage{newfloat}
\newtheorem{theorem}{Theorem}[section]
\newtheorem{mtheorem}{Theorem}
\newtheorem*{theorem*}{Theorem}

\newtheorem*{proposition*}{Proposition}
\newtheorem{lemma}[theorem]{Lemma}
\newtheorem*{lemma*}{Lemma}

\newtheorem*{conjecture*}{Conjecture}
\newtheorem{fact}[theorem]{Fact}
\newtheorem*{fact*}{Fact}

\newtheorem*{hypothesis*}{Hypothesis}

\theoremstyle{definition}
\newtheorem{definition}[theorem]{Definition}
\newtheorem*{definition*}{Definition}

\newtheorem{algorithm}{Algorithm}

\theoremstyle{remark}
\newtheorem{claim}[theorem]{Claim}
\newtheorem*{claim*}{Claim}
\newtheorem{remark}[theorem]{Remark}
\newtheorem{strategy}{Strategy}
\newtheorem*{remark*}{Remark}
\newtheorem{observation}[theorem]{Observation}
\newtheorem*{observation*}{Observation}

\usepackage[
letterpaper,
top=1.2in,
bottom=1.2in,
left=1in,
right=1in]{geometry}
\usepackage{newpxtext} 
\usepackage{textcomp} 
\usepackage[varg,bigdelims]{newpxmath}
\usepackage[scr=rsfso]{mathalfa}
\usepackage{bm} 
\let\mathbb\varmathbb
\usepackage{microtype}

\usepackage{footnotebackref}

\allowdisplaybreaks
\newcommand{\FormatAuthor}[3]{
\begin{tabular}{c}
#1 \\ {\small\texttt{#2}} \\ {\small #3}
\end{tabular}
}

\newcommand{\keywords}[1]{\bigskip\par\noindent{\footnotesize\textbf{Keywords\/}: #1}}

\newcommand{\R}{{\mathbb R}}
\newcommand{\N}{{\mathbb N}}

\newcommand{\abs}[1]{\lvert #1 \rvert}
\newcommand{\floor}[1]{\lfloor #1 \rfloor}

\newcommand{\eps}{\varepsilon}
\newcommand{\defeq}{\coloneqq}
\newcommand{\F}{{\mathbb F}}

\newcommand{\E}{{\mathbb E}}
\newcommand{\1}{\mathbf{1}}
\newcommand{\ip}[1]{\langle #1 \rangle}

\newcommand{\tr}{\mathrm{tr}}

\newcommand{\Bits}{\{0,1\}}
\newcommand{\zo}{\Bits}
\newcommand{\Fits}{\{-1,1\}}

\newcommand{\pE}{\tilde{\E}}
\newcommand{\cH}{\mathcal H}

\newcommand{\cA}{\mathcal A}

\newcommand{\cS}{\mathcal S}
\newcommand{\poly}{\mathrm{poly}}
\newcommand{\val}{\mathrm{val}}

\newcommand{\ceil}[1]{\lceil #1 \rceil}

\newcommand{\mper}{\,.}
\newcommand{\mcom}{\,,}
\newcommand{\Paren}[1]{\left(#1\right)}

\newcommand{\Norm}[1]{\left\lVert#1\right\rVert}
\newcommand{\setQ}{Q}
\newcommand{\Id}{\mathbb{I}}

\newcommand{\polylog}{\mathrm{polylog}}

\newcommand{\Err}{\mathcal{E}}
\newcommand{\Alg}{\mathcal{A}}

\newcommand{\ol}[1]{\overline{#1}}
\newcommand{\wh}[1]{\widehat{#1}}
\newcommand{\wt}[1]{\widetilde{#1}}
\newcommand{\diag}{\operatorname{diag}}
\newcommand{\one}{\vec{1}}

\newcommand{\spn}{\operatorname{span}}
\newcommand{\vol}{\operatorname{vol}}

\newcommand{\Brac}[1]{\left[#1\right]}
\newcommand{\Abs}[1]{\left\lvert#1\right\rvert}
\newcommand{\iprod}[1]{\ip{#1}}
\newcommand{\Iprod}[1]{\left\langle#1\right\rangle}
\newcommand{\bigparen}[1]{\big(#1\big)}
\newcommand{\Bigparen}[1]{\Big(#1\Big)}
\newcommand{\bignorm}[1]{\big\lVert#1\big\rVert}
\newcommand{\Bignorm}[1]{\Big\lVert#1\Big\rVert}
\newcommand{\set}[1]{\{#1\}}
\newcommand{\Set}[1]{\left\{#1\right\}}
\newcommand{\bigset}[1]{\big\{#1\big\}}

\newcommand{\NP}{\mathsf{NP}}
\newcommand{\decomp}{\mathsf{Decomp}}

\newcommand{\BipartiteBeta}{\frac{(k\log n)^{3/2}}{\tau \gamma^2 \eps^{3/2}}}

\begin{document}

\title{Efficient Algorithms for Semirandom Planted CSPs at the Refutation Threshold}
\author{
\begin{tabular}[h!]{ccc}
    \FormatAuthor{Venkatesan Guruswami\thanks{Supported by a Simons Investigator Award and NSF grants CCF-2211972 and CCF-2228287.}}{venkatg@berkeley.edu}{UC Berkeley}
    \FormatAuthor{Jun-Ting Hsieh\thanks{Supported by  NSF CAREER Award \#2047933.}}{juntingh@cs.cmu.edu}{Carnegie Mellon University} \\ \\
    \FormatAuthor{Pravesh K.\ Kothari\thanks{Supported by  NSF CAREER Award \#2047933, Alfred P. Sloan Fellowship and a Google Research Scholar Award.}}{praveshk@cs.cmu.edu}{Carnegie Mellon University}
    \FormatAuthor{Peter Manohar\thanks{Supported in part by an ARCS Scholarship, NSF Graduate Research Fellowship (under grant numbers DGE1745016 and DGE2140739), and NSF CCF-1814603.}}{pmanohar@cs.cmu.edu}{Carnegie Mellon University}
\end{tabular}
} 
\date{\today}

\maketitle
\vspace{-1em}

\begin{abstract}
We present an efficient algorithm to solve semirandom planted instances of any Boolean constraint satisfaction problem (CSP). The semirandom model is a hybrid between worst-case and average-case input models, where the input is generated by (1) choosing an arbitrary planted assignment $x^*$, (2) choosing an arbitrary clause structure, and (3) choosing literal negations for each clause from an arbitrary distribution ``shifted by $x^*$'' so that $x^*$ satisfies each constraint.
For an $n$-variable semirandom planted instance of a $k$-arity CSP, our algorithm runs in polynomial time and outputs an assignment that satisfies all but a $o(1)$-fraction of constraints, provided that the instance has at least $\tilde{O}(n^{k/2})$ constraints.
This matches, up to $\polylog(n)$ factors, the clause threshold for algorithms that solve \emph{fully random} planted CSPs \cite{FeldmanPV15}, as well as algorithms that refute \emph{random and semirandom} CSPs \cite{AllenOW15, AbascalGK21}. Our result shows that despite having worst-case clause structure, the randomness in the literal patterns makes semirandom planted CSPs significantly easier than worst-case, where analogous results require $O(n^k)$ constraints \cite{AroraKK95,FotakisLP16}.

Perhaps surprisingly, our algorithm follows a significantly different conceptual framework when compared to the recent resolution of semirandom CSP refutation. This turns out to be inherent and, at a technical level, can be attributed to the need for \emph{relative} spectral approximation of certain random matrices --- reminiscent of the classical spectral sparsification --- which ensures that an SDP can certify the \emph{uniqueness} of the planted assignment. In contrast, in the refutation setting, it suffices to obtain a weaker guarantee of absolute upper bounds on the spectral norm of related matrices.

\end{abstract}


\thispagestyle{empty}
\setcounter{page}{0}

\clearpage
 \microtypesetup{protrusion=false}
\setcounter{page}{0}
  \tableofcontents{}
  \microtypesetup{protrusion=true}
\thispagestyle{empty}

\clearpage

\pagestyle{plain}
\setcounter{page}{1}

\section{Introduction}
\label{sec:intro}
Four decades of work in computational complexity has uncovered strong hardness results for constraint satisfaction problems (CSPs) such as $k$-SAT that leave only a little room for non-trivial efficient algorithms in the \emph{worst-case}.  Strong hardness of approximation~\cite{Hastad01} essentially rule out (unless $\mathsf{P} = \mathsf{NP}$) any improvement over simply returning a uniformly random assignment when the input instance is \emph{sparse} (i.e., has $m=O(n)$ constraints on $n$ variables). While there is a polynomial time approximation scheme (PTAS)~\cite{AroraKK95} for maximally dense instances (e.g., with $m=O(n^k)$ constraints for $k$-SAT), under the exponential time hypothesis~\cite{ImpagliazzoP01}, we can already rule out polynomial time algorithms for $o(n^k)$ dense instances and more generally, $2^{n^{1-\delta}}$ time algorithms for any $\delta>0$ for $o(n^{k-1})$ dense instances~\cite{FotakisLP16}.

\parhead{Search and refutation in the average-case.} In sharp contrast, in well-studied \emph{average-case} settings, there appears to be significant space for new algorithms and markedly better guarantees for CSPs. CSPs can be studied as two natural problems in such average-case settings: the problem of \emph{refutation} --- where we focus on efficiently finding witnesses of unsatisfiability for models largely supported on unsatisfiable instances, and the problem of \emph{search} --- where our goal is to find an assignment that the model guarantees is \emph{planted} in the instance.

The refutation problem has been heavily investigated in the past two decades. For \emph{fully random} $k$-CSPs with uniformly random clause structure (i.e., which variables appear in each clause) and ``literal pattern'' (i.e., which variables appear negated in each clause), there is a polynomial-time algorithm that, with high probability over the instance, certifies that the instance is unsatisfiable, provided that $m$ is at least $\tilde{O}(n^{k/2})$ \cite{GoerdtL03, CojaGL07, AllenOW15, BarakM16, RaghavendraRS17}. This threshold is far below the $\sim n^{k}$ hardness threshold of \cite{FotakisLP16}. Furthermore, there is  lower bounds in various restricted models \cite{Feige02, BenabbasGMT12, OdonnellW14, MoriW16, BarakCK15, KothariMOW17, FeldmanPV18} provide some evidence that this threshold might be tight for polynomial time algorithms.

The search problem for planted models of CSPs has also received a fair bit of attention. The setting naturally arises in the investigation of \emph{local} one-way functions and pseudorandom generators in cryptography.
Indeed, the security of the well-known one-way function proposed by Goldreich~\cite{Goldreich00} (also conjectured to be a pseudorandom generator~\cite{MosselST06, Applebaum16}) is equivalent to the hardness of recovering a satisfying assignment planted (via a carefully chosen procedure) in a random CSP instance with an appropriate predicate.
This has led to significant research on solving \emph{fully random} planted CSPs~\cite{BarthelHL02, JiaMS07, BogdanovQ09, CojaCF10, FeldmanPV15}.
Specifically, Feldman, Perkins and Vempala~\cite{FeldmanPV15} showed that for \emph{fully random} planted $k$-CSPs with planted assignment $x^*$, there is a polynomial-time algorithm that, with high probability over the instance, recovers the planted assignment $x^*$ \emph{exactly}, provided that the instance has at least $\tilde{O}(n^{k/2})$ constraints. That is, the refutation and search versions have the same clause threshold. 

\parhead{\emph{Beyond} the average-case: semirandom instances} The phenomenal progress in average-case algorithm design notwithstanding, there is a nagging concern that the algorithms so developed rely too heavily on ``brittle'' properties of a specific random model. That is, our methods may have ``overfitted'' to the specific setting thus offering algorithms that only apply in a limited setting. Unfortunately, this fear turns out to be rather well-founded --- natural spectral algorithms for refuting random $k$-CSPs and solving the natural planted variants break down under minor perturbations such as the introduction of a vanishingly small fraction of additional clauses. 

Motivated by such concerns, Blum and Spencer~\cite{BlumS95} and later Feige and Kilian~\cite{FeigeK01, Feige07} introduced \emph{semirandom} models for optimization problems. In semirandom models, the instances are constructed by a combination of benign average-case and adversarial worst-case choices. Algorithms that succeed for such models are naturally ``robust'' to perturbations of the input instance.

For CSPs, a \emph{semirandom} instance is generated by first choosing a ``worst-case'' clause structure and then choosing the literal negation patterns in each clause via some sufficiently random (and thus ``benign'') process. Recent work \cite{AbascalGK21, GuruswamiKM22, HsiehKM23} has shown that in the case of refutation, there are indeed more resilient algorithms that succeed in refuting semirandom instances at the \emph{same} $\tilde{O}(n^{k/2})$ threshold as the \emph{fully random} case. These developments have added new general-purpose new spectral methods based on Kikuchi matrices \cite{WeinAM19, GuruswamiKM22} to our algorithmic arsenal.

\parhead{Semirandom planted problems.} In this work, we make the first step in obtaining algorithms for the \emph{search} variant of CSPs in the semirandom setting. Our main result gives an efficient algorithm for solving semirandom planted CSPs that succeeds in finding the planted assignment whenever the number of constraints exceeds $\tilde{O}(n^{k/2})$  --- the \emph{same} threshold at which polynomial time algorithms exist for the refutation problem for random (and semirandom) instances.

\begin{mtheorem}[Main result, informal \cref{mthm:main}]
\label{mthm:infmain}
There is an efficient algorithm that takes as input a $k$-CSP $\Psi$ and outputs an assignment $x$ with the following guarantee: if $\Psi$ is a semirandom planted $k$-CSP with $m \geq \tilde{O}(n^{k/2})$ constraints, then with high probability over $\Psi$, the output $x$ satisfies $1-o(1)$-fraction of the constraints in $\Psi$.
\end{mtheorem}
We note that in the semirandom setting, it is not possible to efficiently recover an assignment that satisfies \emph{all} of the constraints without being able to do so even when $m = O(n)$\footnote{Achieving this would break a hardness assumption for the search problem analogous to Feige's random $3$-SAT hypothesis for the refutation problem~\cite{Feige02}.}. This is because it is easy to construct a semirandom instance $\psi$ that is the ``union'' of two disjoint instances $\psi_1$ and $\psi_2$, where $\psi_1$ and $\psi_2$ use disjoint sets of $n/2$ variables, but $\psi_1$ only has $m_1 \sim O(n)$ clauses (and $\psi_2$, therefore, contains almost all of the $m \sim n^{k/2}$ clauses). Thus, the guarantee in \cref{mthm:infmain} of satisfying a $1-o(1)$-fraction of constraints is qualitatively the best we can hope for. 

\parhead{Search vs.\ refutation.}
It is natural to compare \cref{mthm:infmain} to the recent resolution of the problem of \emph{refuting} semirandom CSPs \cite{AbascalGK21, GuruswamiKM22, HsiehKM23}. For average-case optimization problems, techniques for refuting random instances can typically be adapted to solving the search problem in the related planted model. This can be formalized in the \emph{proofs to algorithms} paradigm~\cite{BarakS14, FlemingKP19} where spectral/SDP-based refutations can be transformed into ``simple'' (i.e., ''captured" within the low-degree sum-of-squares proof system) efficient certificates of near-uniqueness of optimal solution --- that is, every optimal solution is close to the planted assignment. Unfortunately, this intuition breaks down even in the simplest setting of semirandom $2$-XOR where there can be multiple maximally far-off solutions that satisfy as many (or even more) constraints as the planted assignment. Such departure from uniqueness also breaks algorithms for  recovery~\cite{FeldmanPV15} that rely on the top eigenvector of a certain matrix built from the instance being correlated with the planted assignment. In the semirandom setting, one can build instances where the top eigenspace of such matrices is the span of the multiple optimal solutions and has dimension $\omega(1)$ (searching for a Boolean vector close to the subspace is, in general, hard in super-constant dimensional subspaces).

\parhead{Our key insight.} Our starting point is a new, efficiently checkable certificate of the unique identifiability of the planted solution for noisy planted $k$-XOR (i.e., where each equation in a satisfiable $k$-sparse linear system is corrupted independently with some fixed constant probability) whenever the constraint hypergraph satisfies a certain weak expansion property. For random graphs in case of $2$-XOR (and generalizations to multiple community \emph{stochastic block models}), such certificates (in the form of explicit dual solutions to a semidefinite program) were shown to exist by Abbe and Sandon~\cite{AbbeS15}.

Our certificate naturally yields an efficient algorithm for \emph{exactly} recovering the planted assignment in noisy $k$-XOR instances whenever the constraint hypergraph satisfies a deterministic weak expansion property and has  size exceeding the refutation threshold $\sim n^{k/2}$. Finally, we use expander decomposition procedures to decompose the input constraint hypergraph into pieces that satisfy the above condition. This is done in a manner that further allows us to find a good assignment via a consistent patching scheme to combine solutions across all the pieces in our decomposition.

\subsection{Our semirandom planted model and results}
Before formally stating our results, we define the semirandom planted model that we work with and explain some of the subtleties in the definition. Our model is the natural one that arises if we wish to enforce independent randomness (for each clause) in the literal negations, while still fixing a particular satisfying assignment.
\begin{definition}[$k$-ary Boolean CSPs]
A CSP instance $\Psi$ with a $k$-ary predicate $P \colon \Fits^k \to \Bits$ is a set of $m$ constraints on variables $x_1,\dots,x_n$ of the form $P(\ell(\vec{C})_1 x_{\vec{C}_1}, \ell(\vec{C})_2 x_{\vec{C}_2}, \ldots, \ell(\vec{C})_k x_{\vec{C}_k}) = 1$. Here, $\vec{C}$ ranges over a collection $\vec{\cH}$ of \emph{scopes}\footnote{We additionally allow $\vec{\cH}$ to be a multiset, i.e., that multiple clauses can contain the same ordered set of variables.} (a.k.a.\ clause structure) of $k$-tuples of $n$ variables and $\ell(\vec{C}) \in \Fits^k$ are ``literal negations'', one for each $\vec{C}$ in $\vec{\cH}$. We let $\val_{\Psi}(x)$ denote the fraction of constraints satisfied by an assignment $x \in \Fits^n$, and we define the \emph{value} of $\Psi$, $\val(\Psi)$, to be $\max_{x \in \Fits^n} \val_{\Psi}(x)$.
\end{definition}

\begin{definition}[Semirandom planted $k$-ary Boolean CSPs]
\label{def:semirandomcsp}
Let $P \colon \Fits^k \to \Bits$ be a predicate. We say that a distribution $Q$ over $\Fits^k$ is a \emph{planting distribution for $P$} if $\Pr_{y \gets Q}[P(y) = 1] = 1$.

We say that an instance $\Psi$ with predicate $P$ is a \emph{semirandom planted instance} with \emph{planting distribution} $Q$ if it is sampled from a distribution $\Psi(\vec{\cH}, x^*, Q)$ where
\begin{enumerate}[(1)]
    \item the scopes $\vec{\cH} \subseteq [n]^k$ and planted assignment $x^* \in \Fits^n$ are arbitrary, and
    
    \item $\Psi(\vec{\cH}, x^*, Q)$ is defined as follows: for each $\vec{C} \in \vec{\cH}$, sample literal negations $\ell(\vec{C}) \gets Q(\ell(\vec{C}) \odot x^*_{\vec{C}})$, where ``$\odot$'' denotes the element-wise product of two vectors.
    That is, $\Pr[\ell(\vec{C}) = \ell] = Q(\ell \odot x^*_{\vec{C}})$ for each $\ell \in \Fits^k$.
    Then, add the constraint $P(\ell(\vec{C})_1 x_{\vec{C}_1}, \ell(\vec{C})_2 x_{\vec{C}_2}, \ldots, \ell(\vec{C})_k x_{\vec{C}_k}) = 1$ to $\Psi$.
\end{enumerate}
Notice that because $Q$ is supported only on satisfying assignments to $P$, it follows that if $\Psi \gets \Psi(\vec{\cH}, x^*, Q)$, then $x^*$ satisfies $\Psi$ with probability $1$.
\end{definition}
A (fully) random planted CSP, e.g., as defined in \cite{FeldmanPV15}, is generated by first sampling $\vec{\cH} \gets [n]^k$ uniformly at random, and then sampling $\Psi \gets \Psi(\vec{\cH}, x^*, Q)$. The difference in the semirandom planted model is that we allow $\vec{\cH}$ to be \emph{worst case}.

Notice that in \cref{def:semirandomcsp}, there are some choices of $Q$ for which the planted instance becomes easy to solve. In the case of, e.g., $3$-SAT, one could set the planting distribution $Q$ to be uniform over all $7$ satisfying assignments, which results in the literal negations in each clause being chosen uniformly conditioned on $x^*$ satisfying the clause. However, by simply counting how many times the variable $x_i$ appears negated versus not negated and taking the majority vote, we recover $x^*$ with high probability~\cite{BarthelHL02,JiaMS07} (see \cref{sec:1xor}).

Instead of sampling clauses uniformly from all those satisfied by $x^*$, one can create more challenging distributions, e.g., ones where true and false literals appear in equal proportion. Such distributions are termed ``quiet plantings'' and have been studied extensively~\cite{JiaMS07,KrzakalaZ09,CojaCF10,KrzakalaMZ12}. Our semirandom model follows definitions in \cite{FeldmanPV15,FeldmanPV18} and is a general planted model with respect to a \emph{planting distribution $Q$}, which unifies various plantings studied in the past.

Unlike in the case of random planted CSPs, we cannot hope to recover the planted assignment $x^*$ exactly in the semirandom setting. Indeed, the scopes $\vec{\cH}$ may not use some variable $x_i$ at all, and so we cannot hope to recover $x^*_i$! Thus, our goal is instead to recover an assignment $x$ that has nontrivially large value, ideally value $1 - \eps$ for arbitrarily small $\eps$. Our main result, stated formally below, gives an algorithm to accomplish this task.
\begin{mtheorem}[Formal \cref{mthm:infmain}]
\label{mthm:main}
Let $k \in \N$ be constant. There is a polynomial-time algorithm that takes as input a $k$-CSP $\Psi$ and outputs an assignment $x$ with the following guarantee. If $\Psi$ is a semirandom planted $k$-CSP with $m \geq c^k n^{k/2} \cdot \frac{\log^3 n}{\eps^9}$ constraints drawn from $\Psi(\vec{\cH}, x^*, Q)$, then with probability $1 - 1/\poly(n)$ over $\Psi$, the output $x$ of the algorithm has $\val_{\Psi}(x) \geq 1 - \eps$. Here, $c$ is a universal constant.

In particular, setting $\eps = 1/\polylog(n)$, if $m \geq \tilde{O}(n^{k/2})$, then with high probability over $\Psi \gets \Psi(\vec{\cH}, x^*, Q)$, the algorithm outputs $x$ with $\val_{\Psi}(x) \geq 1 - o(1)$.
\end{mtheorem}

\cref{mthm:main} shows that one can \emph{nearly} solve a semirandom planted $k$-CSP at the same $\tilde{O}(n^{k/2})$ threshold as done in the random case \cite{FeldmanPV15}, matching the same $\tilde{O}(n^{k/2})$ threshold as for semirandom refutation \cite{AbascalGK21, GuruswamiKM22, HsiehKM23}. However, as explained earlier (and will be discussed further in \cref{sec:techniques}), there are several unanticipated technical hurdles to overcome in the semirandom planted setting that are not present in the semirandom refutation setting, and this causes many of the natural approaches that ``springboard off'' the refutation case to fail. Curiously enough, for the special case of $k = 2$ there \emph{is} a simple reduction from search to refutation for the case of $2$-XOR, which we will describe in \cref{sec:2xorref}, but the same approach for $k$-XOR encounters a hardness barrier for $k \geq 3$, as we will discuss in \cref{sec:challenges-for-kxor}.

\cref{mthm:main} also breaks Goldreich's candidate pseudorandom generators~\cite{Goldreich00} and its variants~\cite{Applebaum16},\footnote{Goldreich's original PRG is essentially a planted $k$-CSP with a Boolean predicate $P$ on a random hypergraph, containing both $P$ and $\neg P$ constraints.} when they have $\tilde{\Omega}(n^{k/2})$ stretch and \emph{any} $k$-hypergraph (not just a random one). In fact, not only does \cref{mthm:main} break the PRG, it also gives an algorithm that nearly \emph{inverts} it.

\parhead{Noisy planted $k$-XOR.}
Similar to work on random planted CSPs \cite{FeldmanPV15} and the refutation setting \cite{AllenOW15, RaghavendraRS17, AbascalGK21, GuruswamiKM22, HsiehKM23}, our proof of \cref{mthm:main} goes through a reduction to noisy $k$-XOR. Our algorithm achieves very strong guarantees in the noisy $k$-XOR case, as we now explain. We define the noisy $k$-XOR model below and then state our result.

\begin{definition}[Noisy planted $k$-XOR]
\label{def:semirandomxor}
Let $\cH \subseteq \binom{[n]}{k}$ be a $k$-uniform hypergraph on $n$ vertices, let $x^* \in \Fits^n$, and let $\eta \in [0,1/2)$. Let $\psi(\cH, x^*, \eta)$ denote the distribution on $k$-XOR instances over $n$ variables $x_1, \dots, x_n \in \Fits$ obtained by, for each $C \in \cH$,
adding the constraint $\prod_{i \in C} x_i = \prod_{i \in C} x^*_i$ with probability $1 - \eta$, and otherwise adding the constraint $\prod_{i \in C} x_i = -\prod_{i \in C} x^*_i$. In the latter case, we say that the constraint $C$ is \emph{corrupted} or \emph{noisy}.

We call $\psi$ a \emph{noisy planted} $k$-XOR instance if it is sampled from $\psi(\cH, x^*, \eta)$, for some $\cH$, $x^*$, and $\eta$; the hypergraph $\cH$ is the constraint hypergraph, $x^*$ is the planted assignment, and $\eta$ is the noise parameter. Furthermore, we let $\Err_{\psi} \subseteq \cH$ denote the (unknown) set of corrupted constraints.
\end{definition}

\begin{mtheorem}[Algorithm for noisy $k$-XOR]
\label{mthm:algxor}
Let $\eta \in [0, 1/2)$, let $k, n \in \N$, and let $\eps \in (0, 1)$. Let $m \geq c n^{k/2} \cdot \frac{k^4 \log^3 n}{\eps^5 (1-2\eta)^4}$ for a universal constant $c$. There is a polynomial-time algorithm $\Alg$ that takes as input a $k$-XOR instance $\psi$ with constraint hypergraph $\cH$ and outputs two disjoint sets $\Alg_1(\cH), \Alg_2(\psi) \subseteq \cH$ with the following guarantees: \begin{inparaenum}[(1)] \item for any instance $\psi$ with $m$ constraints, $\abs{\Alg_1(\cH)} \leq \eps m$ and $\Alg_1(\cH)$ only depends on $\cH$, and \item for any $x^* \in \Fits^n$ and any $k$-uniform hypergraph $\cH$ with at least $m$ hyperedges, with probability at least $1 - 1/\poly(n)$ over $\psi \gets \psi(\cH, x^*, \eta)$, it holds that $\Alg_2(\psi) = \Err_{\psi} \cap (\cH \setminus \Alg_1(\cH))$.\end{inparaenum}
\end{mtheorem}

In words, the algorithm discards a small number of constraints, and among the constraints that are not discarded, correctly identifies all (and only) the corrupted constraints. 
In particular, the subinstance obtained by discarding the $\lesssim (\eps + \eta)m$ constraints $\Alg_1(\cH) \cup \Alg_2(\psi)$ is satisfiable (and a solution can be found by Gaussian elimination).
Thus, \cref{mthm:algxor} immediately implies that for $k$-XOR, the $\NP$-hard task of deciding if $\psi$ has value $\geq 1 - \eta$ or $\leq \frac{1}{2} + \eta$ is actually \emph{easy} if $\psi$ has $\sim n^{k/2}$ 
constraints (far below the $\sim n^{k}$-hardness of \cite{FotakisLP16}), provided that the $\eta$-fraction of corrupted constraints in the ``yes'' case are a \emph{randomly chosen subset} of the otherwise arbitrary constraints.

\parhead{Exact vs.\ approximate recovery.}
As alluded to above, the guarantees of \cref{mthm:algxor} are much stronger: not only can we find a good assignment to $\psi$, we can break the constraints into two parts, a small fraction, $\Alg_1(\cH)$, where we are unable to determine the corrupted constraints,\footnote{Note that discarding a small fraction of constraints is necessary in the semirandom setting, as $\psi$ may contain many disconnected constant-size subinstances where it is not possible, even information-theoretically, to exactly identify the corrupted constraints with $1-o(1)$ probability.} and a large fraction, $\cH \setminus \Alg_1(\cH)$, where we can determine \emph{exactly} all of the corrupted constraints, $\Alg_2(\psi)$. Moreover, this partition depends only on the hypergraph $\cH$ and is \emph{independent of the noise}. We remark that it is not immediately obvious that this guarantee is achievable even for exponential-time algorithms, as $x^*$ may not be the globally optimal assignment with constant probability.
This strong guarantee of \cref{mthm:algxor} is in fact required for the reduction from \cref{mthm:main} to \cref{mthm:algxor}; the weaker (and more intuitive) guarantee of approximate recovery --- obtaining an assignment of value $1 - \eta - o(1)$ for the noisy XOR instance --- is insufficient for the reduction. 

One can view \cref{mthm:algxor} as an algorithm that extracts almost all the information about the planted assignment $x^*$ encoded by the instance $\psi$. Indeed, notice that even if $\eta = 0$, the instance $\psi$ only determines $x^*$ ``up to a linear subspace.''\footnote{A $k$-XOR constraint $x_{C_1} \cdots x_{C_k} = b_C \in \Fits$ can be equivalently written as a linear equation $x_{C_1}' + \cdots + x_{C_k}' = b_C'$ over $\F_2$, where we map $+1$ to $0$ and $-1$ to $1$.}
Namely, if we let $y \in \Fits^n$ be any solution to the system of constraints $\prod_{i \in C} y_i = 1$ for $C \in \cH$, then $y \odot x^*$ is also a planted assignment for $\psi$: formally, $\psi(\cH, x^*, \eta) = \psi(\cH, y \odot x^*, \eta)$ as distributions. So, aside from the $\eps m$ constraints that are discarded, with high probability over $\psi$ the algorithm determines the uncorrupted right-hand sides $\prod_{i \in C} x^*_i$ for every remaining constraint, which is all the information about the planted assignment $x^*$ encoded in the remaining constraints.

\parhead{The importance of relative spectral approximation.}
As a key technical ingredient in the algorithm, we uncover a \emph{deterministic} condition --- relative spectral approximation of the Laplacian of a graph (associated with the input instance) by a certain correlated random sample from it --- which when satisfied implies uniqueness of the SDP solution (\cref{lem:sdpuniqueness}). In  \cref{lem:spectral-sparsification} and \cref{lem:spectral-sparsification-odd}, we establish such spectral approximation guarantees. 

This spectral approximation property is the key ingredient in our certificate of unique identifiability of the planted assignment in a noisy $k$-XOR instance (see \cref{sec:sparsificationtechniques} for details) and allows us to \emph{exactly} recover the planted assignment for $2$-XOR instances where the constraint graph $G$ is a weak spectral expander (i.e., spectral gap $\gg 1/\poly \log n$) (\cref{lem:sdpuniqueness}), and forms the backbone of our final algorithm. We note that our spectral approximation condition can be seen as an analog of (and is, in fact, stronger than) the related spectral norm upper bound property that underlie the refutation algorithm of \cite{AbascalGK21}. 


This process of extracting a ``deterministic property of random instances sufficient for the analysis'' is an important conceptual theme underlying recent progress on semirandom optimization, and manifests as, e.g., the notion of ``butterfly degree'' in \cite{AbascalGK21}, ``hypergraph regularity'' or spreadness in \cite{GuruswamiKM22} in the context of semirandom CSP refutation, and biclique number bounds in the context of planted clique~\cite{BuhaiKS23}.
\section{Technical Overview}
\label{sec:techniques}
In this section, we give an overview of the proof of \cref{mthm:algxor} and our algorithm for noisy planted $k$-XOR. We defer discussion of the reduction from general $k$-CSPs to $k$-XOR used to obtain \cref{mthm:main} to \cref{sec:reductiontoxor}. There, we explain the additional challenges encountered in the semirandom case as compared to the random case \cite[Section 4]{FeldmanPV15}. Somewhat surprisingly, the reduction is complicated and quite different from the random planted case or even the semirandom refutation setting, where the reduction to XOR is straightforward.

We now explain \cref{mthm:algxor}. As is typical in algorithm design for $k$-XOR, the case when $k$ is even is considerably simpler than when $k$ is odd. For the purpose of this overview, we will focus mostly on the even case, and only briefly discuss the additional techniques for odd $k$ in \cref{sec:oddk}.

\parhead{Notation.} Throughout this paper, given a $k$-XOR instance $\psi$ on hypergraph $\cH \subseteq \binom{[n]}{k}$ with $m = |\cH|$ and right-hand sides $\{b_C\}_{C\in \cH}$, we define $\psi(x) \defeq \sum_{C\in \cH} b_C \prod_{i\in C} x_i$ to be a degree-$k$ polynomial mapping $\Fits^n \to [-m,m]$.
We note that $\val_{\psi}(x) = \frac{1}{2} + \frac{1}{2m}\psi(x) \in [0,1]$ is the fraction of constraints in $\psi$ satisfied by $x$.
Moreover, we will write $x_C \defeq \prod_{i\in C} x_i$.

Unless otherwise stated, we will use $\phi$ to denote a $2$-XOR instance and $\psi$ to denote a $k$-XOR instance for any $k \geq 2$.

We note that for even arity $k$-XOR, we have $\val_{\psi}(x) = \val_{\psi}(-x)$, and so it is only possible for the optimal solution to be unique \emph{up to a global sign}. We will abuse terminology and say that $x^*$ is the unique optimal assignment if $\pm x^*$ are the only optimal assignments, and we will say that we have recovered $x^*$ exactly if we obtain one of $\pm x^*$.
\subsection{Approximate recovery for \texorpdfstring{$2$}{2}-XOR from refutation}
\label{sec:2xorref}

First, let us focus on the case of $k = 2$, the simplest case, and let us furthermore suppose that we only want to achieve the weaker goal of recovering an assignment of value $1 - \eta - o(1)$. (Note that we do need the stronger guarantee of \cref{mthm:algxor} to solve general planted CSPs in \cref{mthm:main}.)

For $2$-XOR, this goal is actually quite straightforward to achieve using $2$-XOR refutation as a blackbox. Let us represent the $2$-XOR instance $\phi$ as a graph $G$ on $n$ vertices, along with right-hand sides $b_{ij}$ for each edge $(i,j) \in E$. Recall that we have $b_{ij} = x^*_i x^*_j$ with probability $1 - \eta$, and $b_{ij} = -x^*_i x^*_j$ otherwise.
Note that by concentration, $\val_{\phi}(x^*) = 1 - \eta \pm o(1)$ with high probability.

We now make the following observation. Let us suppose that we sample the noise in two steps: first, we add each $(i,j) \in E$ to a set $E'$ with probability $2\eta$ independently; then for each $(i,j) \in E'$ we set $b_{ij}$ to be uniformly random from $\Fits$.
Using known results for semirandom $2$-XOR refutation, it is possible to certify, via an SDP relaxation, that no assignment $x$ can satisfy (or violate) more than $\frac{1}{2} + o(1)$ fraction of the constraints in $E'$.

Thus, we can simply solve the SDP relaxation for $\phi$ and obtain a degree-$2$ pseudo-expectation $\pE$ in the variables $x_1, \dots, x_n$ over $\Fits^n$ that maximizes $\phi(x)$.
Let $\phi_{E'}$ be the subinstance containing only the constraints in $E'$, and let $\phi_{E\setminus E'}$ be the subinstance containing only the constraints in $E \setminus E'$, which are uncorrupted.
We have $\pE[\val_{\phi}(x)] \geq 1 - \eta - o(1)$, and the guarantee of refutation implies that $\pE[\val_{\phi_{E'}}(x)] \leq \frac{1}{2} + o(1)$. As $\val_{\phi}(x) = (1 - 2 \eta) \cdot \val_{\phi_{E \setminus E'}}(x) + 2 \eta \cdot \val_{\phi_{E'}}(x)$, we therefore have that $\pE[\val_{\phi_{E\setminus E'}}(x)] \geq 1-o(1)$, i.e., $\pE$ satisfies $1-o(1)$ fraction of the constraints in $E \setminus E'$.
Then, applying the standard Gaussian rounding, we obtain an $x$ that satisfies $1 - \sqrt{o(1)}$ fraction of the constraints in $E \setminus E'$ and thus has value $\val_{\phi}(x) \geq 1 - \eta - o(1)$ (as any $x$ must satisfy at least $\frac{1}{2} - o(1)$ fraction of the constraints in $E'$, with high probability over the noise).

One interesting observation is that in the above discussion, we can additionally allow $E'$ to be an \emph{arbitrary} subset of $E$ of size $2 \eta m$. Indeed, this is because the rounding only ``remembers'' that $\pE[\val_{\phi_{E \setminus E'}}(x)]$ has value $1 - o(1)$. As we shall see shortly, this is the key reason that the reduction breaks down for $k$-XOR.

\subsection{The challenges for \texorpdfstring{$k$}{k}-XOR and our strategy}
\label{sec:challenges-for-kxor}

Unfortunately, the natural blackbox reduction to refutation given in \cref{sec:2xorref} does not generalize to $k$-XOR for $k \geq 3$.
Following the approach described in the previous section, given a $k$-XOR instance $\psi$, one can solve a sum-of-squares SDP and obtain a pseudo-expectation $\pE$ where $\pE[\val_{\psi}(x)] \geq 1 - \eta - \delta$ and $\pE[\val_{\psi_{E\setminus E'}}(x)] \geq 1-\delta$ as before, where $\delta \sim 1/\polylog(n)$ when $m \gtrsim n^{k/2}$, due to the guarantees of refutation algorithms~\cite{AbascalGK21}.
However, unlike $2$-XOR where we have Gaussian rounding, for $k$-XOR there is no known rounding algorithm that takes a pseudo-expectation $\pE$ with $\pE[\val_{\psi_{E\setminus E'}}(x)] \geq 1-\delta$ and outputs an assignment $x$ such that $\val_{\psi_{E\setminus E'}}(x) \geq 1-f(\delta)$, for some $f(\cdot)$ such that $f(\delta)\to 0$ as $\delta\to 0$.
In fact, if we only ``remember'' that $\psi_{E\setminus E'}$ has value $1-\delta$, then it is $\NP$-hard to find an $x$ with value $> 1/2 + \delta$ even when $\delta = n^{-c}$ for some constant $c>0$, assuming a variant of the Sliding Scale Conjecture~\cite{BellareGLR93}\footnote{Note that we do need the Sliding Scale Conjecture, as the hardness shown in \cite{MoshkovitzR10} is not strong enough; it only proves hardness for $\delta \geq (\log \log n)^{-c}$, whereas we have $\delta \sim 1/\polylog(n)$.} (see e.g.~\cite{MoshkovitzR10, Moshkovitz15} for more details).

As we have seen, while semirandom $k$-XOR refutation allows us to efficiently approximate and certify the \emph{value} of the planted instance, the challenge lies in the \emph{rounding} of the SDP, where the goal is to recover an assignment $x$. This is a technical challenge that does not arise in the context of CSP refutation, as there we are merely trying to bound the value of the instance. As a result, new ideas are required to address this challenge.

\medskip
\parhead{Reduction from $k$-XOR to $2$-XOR for even $k$.}
One could still consider the following natural approach. For simplicity, let $k = 4$. Given a $4$-XOR instance $\psi$, we can write down a natural and related $2$-XOR instance $\phi$, as follows.
\begin{definition}[Reduction to $2$-XOR]
\label{def:reductionto2xor}
Let $\psi$ be a $4$-XOR instance, and let $\phi$ be the $2$-XOR defined as follows. The variables of $\phi$ are $y_{\{i,j\}}$ and correspond to \emph{pairs} of variables $\{x_i, x_j\}$, and for each constraint $x_i x_j x_{i'} x_{j'} = b_{i, j, i', j'}$ in $\psi$, we split $\{i,j,i',j'\}$ into $\{i,j\}$ and $\{i',j'\}$ arbitrarily and add a constraint $y_{\{i,j\}} y_{\{i', j'\}} = b_{i, j, i', j'}$ to $\phi$. See \cref{fig:4xor-example} for an example. This reduction easily generalizes to $k$-XOR for any even $k$.
\end{definition}

\begin{figure}[ht!]
    \centering
    \includegraphics[width=0.9\textwidth]{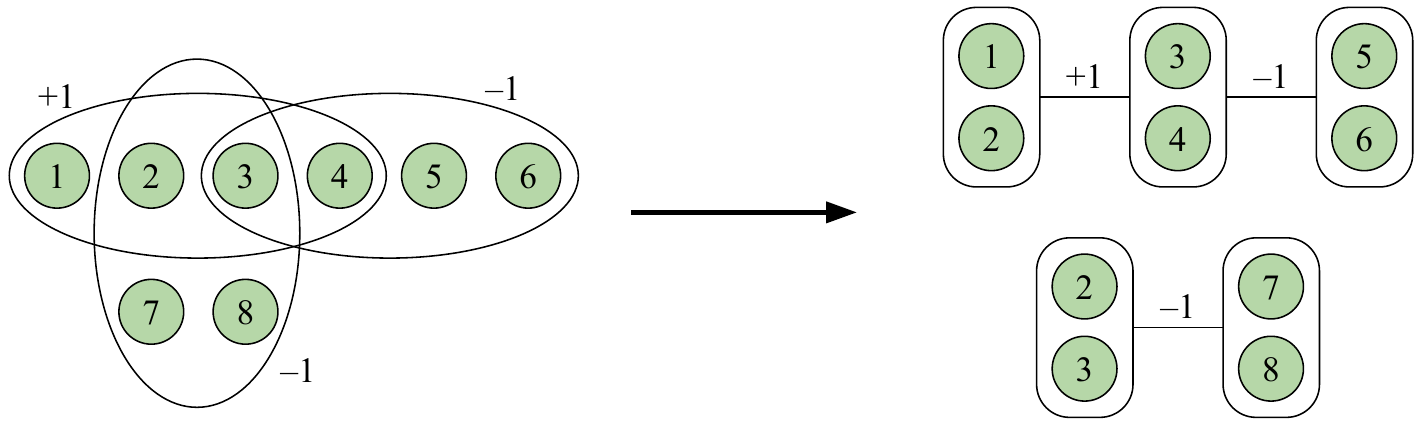}
    \caption{An example of the 2-XOR instance $\phi$ from a
    $4$-XOR instance $\psi$.}
    \label{fig:4xor-example}
\end{figure}

By following the approach for $2$-XOR described in \cref{sec:2xorref}, we can recover an assignment $y$ that satisfies $1 - \eta - o(1)$ fraction of the constraints in $\phi$. However, we need to recover an assignment $x$ to the original $k$-XOR $\psi$, and it is quite possible that while $y$ is a good assignment to $\phi$, \emph{it is not} close to $x^{\otimes 2}$ for \emph{any} $x \in \Fits^n$. If this happens, we will be unable to recover a good assignment to the $4$-XOR instance $\psi$.

The key reason that this simple idea fails is because, unlike for random noisy XOR, the assignment $y$ recovered is \emph{not} necessarily unique, and we cannot hope for it to be in the semirandom setting! For random noisy XOR, one can argue that with high probability, $y$ will be equal to ${x^*}^{\otimes 2}$, and then we can immediately decode and recover $x^*$ up to a global sign, i.e., we recover $\pm x^*$. But for semirandom instances, the situation can be far more complex.

\parhead{Approximate $2$-XOR recovery does not suffice for $4$-XOR.}
When constructing the $2$-XOR instance $\phi$ from the $4$-XOR $\psi$ (\cref{def:reductionto2xor}), it may be the case that $\phi$ can be partitioned into multiple disconnected clusters (or have very few edges across different clusters), even when the hypergraph $\cH$ of $\psi$ is connected; see \cref{fig:4xor-example} for example.
By the algorithm described in \cref{sec:2xorref}, we can get an assignment $y$ that satisfies $1-\eta-o(1)$ fraction of the constraints within each cluster.

The main challenge is to combine the information gathered from each cluster to recover an assignment $x$ for the original $4$-XOR $\psi$. Unfortunately, we do not know of a way to obtain a good assignment $x$ based solely on the guarantee that $y$ satisfies $1-\eta-o(1)$ fraction of constraints in each cluster. The issue occurs because the same variable $i \in [n]$ can appear in different clusters, e.g., $y_{\{1,2\}}$ and $y_{\{2,3\}}$ lie in different clusters in \cref{fig:4xor-example}, and the recovered assignments in each cluster may implicitly choose different values for $x_i$ because of the noise.
Indeed, even if the local optimum is consistent with $x^*$, there can still be multiple ``good'' assignments that achieve $1-\eta-o(1)$ value on the subinstance restricted to a cluster. So, unless the SDP can certify unique optimality of $x^*$, standard rounding techniques such as Gaussian rounding will merely output a ``good'' $y$, which may be inconsistent with $x^*$ and thus can choose inconsistent values of $x_i$ across the different clusters.

\parhead{Exact $2$-XOR recovery implies exact $4$-XOR recovery.} This leads to our main insight: if the subinstance of $\phi$ admits a \emph{unique} local optimal assignment $y^*$ (restricted to the cluster) that matches the planted assignment up to a sign, i.e., $y_{\{i,j\}}^* = \pm x_i^* x_j^*$, then for each edge in the cluster we know $y_{\{i,j\}}^* y_{\{i',j'\}}^* = x_i^* x_j^* x_{i'}^* x_{j'}^*$, and so the local constraints that are violated must be exactly the corrupted ones.
Moreover, if the SDP can certify the uniqueness of the local optimal assignment for a cluster, then the SDP solution will be a \emph{rank $1$ matrix} $y^* y^{*\top}$, and so we can precisely identify which constraints in $\phi$ are corrupted.
By repeating this for every cluster, we can identify all corrupted constraints in the original $4$-XOR $\psi$ (except for the small number of ``cross cluster'' edges), and thus achieve the guarantee stated in \cref{mthm:algxor}.

\parhead{The general algorithmic strategy.} The above discussion suggests that given a $k$-XOR instance $\psi$, we should first construct the $2$-XOR $\phi$, and then decompose the constraint graph $G$ of $\phi$ into pieces in some particular way so that the induced local instances have unique solutions. Namely, the examples suggest the following algorithmic strategy.
\begin{mdframed}
\begin{strategy}[Algorithm Blueprint for even $k$]
\label{strat:alg}
Given a noisy $k$-XOR instance $\psi$ with planted assignment $x^*$ and $m$ constraints, we do the following:
\begin{enumerate}[(1)]
\item Construct the $2$-XOR instance $\phi$ described in \cref{def:reductionto2xor}, which is a noisy $2$-XOR on $n^{k/2}$ variables with planted assignment $y^*$. Moreover, there is a one-to-one mapping between constraints in $\phi$ and $\psi$.

\item Let $G$ be the constraint graph of $\phi$. Decompose $G$ into subgraphs $G_1, \dots, G_T$ while only discarding a $o(1)$-fraction of edges such that each subgraph $G_i$ satisfies ``some property''. For each subgraph $G_i$, we define $\phi_i$ to be the subinstance of $\phi$ corresponding to the constraints in $G_i$. The goal is to identify a local property that the $G_i$'s satisfy so that \begin{inparaenum}[(1)]
\item we can perform the decomposition efficiently, and
\item for each subinstance $\phi_i$, we can ``recover $y^*$ locally'', i.e., we can find an assignment $y^{(i)}$ to the $2$-XOR instance $\phi_i$ that is consistent with the planted assignment $y^*$.
\end{inparaenum}
\item As each $y^{(i)}$ is consistent with $y^*$, the constraints in $\phi_i$ violated by $y^{(i)}$ must be precisely the corrupted constraints in $\phi_i$. Hence, for the constraints that appear in one of the $\phi_i$'s, we have determined exactly which ones are corrupted.

\item We have thus determined, for all but $o(m)$ constraints, precisely which ones are corrupted in the original $k$-XOR instance $\psi$. (Note that this is the \emph{stronger} guarantee that we achieve in \cref{mthm:algxor}.) By discarding the corrupted constraints along with the $o(m)$ constraints where we ``give up'', we thus obtain a system of $k$-sparse linear equations with $m(1- \eta - o(1))$ equations that has at least one solution (namely $x^*$), and so by solving it we obtain an $x$ with $\val_{\psi}(x) \geq 1 - \eta - o(1)$.
\end{enumerate}
\end{strategy}
\end{mdframed}

\subsection{Information-theoretic exact recovery from relative cut approximation}
Following \cref{strat:alg}, the first technical question to now ask is: given a noisy $2$-XOR instance $\phi$ with $n$ variables, $m \gg n$ constraints, and planted assignment $x^*$, what conditions do we need to impose on the constraint graph $G$ so that we can recover $x^*$ (up to a sign) exactly? As a natural first step, we investigate what conditions are required so that we can accomplish this \emph{information-theoretically}.
\begin{fact}
\label{fact:infotheoretic}
    Let $G = (V,E_G)$ be an $n$-vertex graph, and let $H = (V, E_H)$ be a subgraph of $G$ where $E_H \subseteq E_G$.
    Let $L_G, L_H$ be the \emph{unnormalized} Laplacians of $G$ and $H$.
    Consider a noisy planted 2-XOR instance $\phi$ on $G$ with planted assignment $x^* \in \Fits^n$ (\cref{def:semirandomxor}), and suppose $E_H$ is the set of corrupted edges.
Suppose that for every $x \in \Fits^n \setminus \{\one, -\one\}$, it holds that $x^{\top} L_H x < \frac{1}{2} x^{\top} L_G x$. Then, $x^*$ and $-x^*$ are the only two optimal assignments to $\phi$.
\end{fact}
Note that the condition $x^{\top} L_H x < \frac{1}{2} x^{\top} L_G x$ for $x \notin \{\one, -\one\}$ implies that $G$ is connected, as otherwise $L_G$ has a kernel of dimension $\geq 2$, which would contradict this assumption.
\begin{proof}
Let $x \in \Fits^n$ be any assignment. We wish to show that $\phi(x)$ is uniquely maximized when $x = x^*, -x^*$. We observe that 
\begin{flalign*}
&\phi(x) = \sum_{(i,j) \in E_G} x_i x_j b_{ij} = \sum_{(i,j) \in E_G} x_i  x_j x^*_i x^*_j - 2 \sum_{(i,j) \in E_H} x_i  x_j x^*_i x^*_j \enspace.
\end{flalign*}
Hence, by replacing $x$ with $x \odot x^*$, without loss of generality we can assume that $x^* = \one$. Now, let $D_G, D_H$ and $A_G, A_H$ be the degree and adjacency matrices of $G$ and $H$, so that $L_G = D_G - A_G$ and $L_H = D_H - A_H$. We thus have that
\begin{flalign*}
2\phi(x) &= x^\top A_G x - 2 x^\top A_H x = x^{\top} (D_G - 2 D_H) x - x^{\top} (L_G - 2 L_H) x \\
&= 2(\abs{E_G} - 2\abs{E_H}) - x^{\top} (L_G - 2 L_H) x \enspace.
\end{flalign*}
By assumption, if $x \in \Fits^n$ and $x \ne \one, -\one$, then we have that $x^{\top} (L_G - 2 L_H) x > 0$, which implies that $\phi(x) < \phi(\one)$, and finishes the proof.
\end{proof}

\cref{fact:infotheoretic} shows that if we can argue that $x^{\top}  L_H x < \frac{1}{2} x^{\top} L_G x$ for every $x \in \Fits^n \setminus \{\one, -\one\}$, then at least information-theoretically we can uniquely determine $x^*$. Observe that if we view $x$ as the signed indicator vector of a subset $S \subseteq [n]$, then $x^{\top} L_G x = E_G(S,\bar{S})$, the number of edges in $G$ crossing the cut defined by $S$, and similarly for $x^{\top} L_H x$. So, one can view the condition in \cref{fact:infotheoretic} as saying that the subgraph $H$ needs to be a (one-sided) cut sparsifier of $G$, i.e., it needs to roughly preserve the size of all cuts in $G$. The following relative cut approximation result of Karger~\cite{Karger94} shows that this will hold with high probability when $H$ is a randomly chosen subset of $G$, provided that the minimum cut in $G$ is not too small.
\begin{lemma}[Relative cut approximation \cite{Karger94}]
\label{lem:cutsparsification}
    Let $\eta \in (0, 1)$.
    Suppose an $n$-vertex graph $G$ has min-cut $c_{\min} \geq \frac{12\log n}{\eta}$, and suppose $H$ is a subgraph of $G$ by selecting each edge with probability $\eta$. Then, with probability $1-o(1)$,
    \begin{equation*}
        (1-\delta) x^{\top} L_G x \leq \frac{1}{\eta} \cdot x^{\top} L_H x \leq (1+\delta) x^{\top} L_G x \mcom \quad \text{for all $x \in \Fits^n$}
    \end{equation*}
    for $\delta = \sqrt{\frac{12 \log n}{\eta c_{\min}}}$.
\end{lemma}
With \cref{lem:cutsparsification} and \cref{fact:infotheoretic} in hand, we now have at least an information-theoretic algorithm with the same guarantees as in \cref{mthm:algxor}. We follow the strategy highlighted in \cref{strat:alg}. To decompose the graph $G$, we recursively find a min cut and split if it is below the threshold in \cref{lem:cutsparsification}. Notice that this discards at most $O(n \log n) = o(m)$ constraints (for $m \gg n \log n$), and these are precisely the constraints that we ``give up'' on and do not determine which ones are corrupted. Then, with high probability the local optimal assignment is consistent with $x^*$, and so locally we have learned \emph{exactly} which constraints are corrupted.
Hence, we have produced two sets of constraints: $E_1$, the $o(1)$-fraction of edges discarded during the decomposition, and $E_2 = (G \setminus E_1) \cap \Err_{\phi}$, which is exactly the set of corrupted constraints after discarding $E_1$. We note that it is a priori not obvious that this is achievable even for an \emph{exponential-time} algorithm, 
 as even though the $2^n$-time brute force algorithm will find the best assignment $x$ to $\phi$, it may not necessarily be $x^*$, and so the set of constraints violated by the globally optimal assignment might not be $\Err_{\phi}$.

\subsection{Efficient exact recovery from relative spectral approximation}
\label{sec:sparsificationtechniques}
Information-theoretic uniqueness implies that the planted assignment $x^*$ is the unique optimal assignment. But can we efficiently recover $x^*$? One natural approach is to simply solve the basic SDP relaxation of $\phi$: for $X \in \R^{n \times n}$, maximize $\phi(X) \defeq \sum_{(i,j) \in G} X_{ij} b_{ij}$ subject to $X \succeq 0$, $X = X^{\top}$, and $\diag(X) = \Id$. If the optimal SDP solution is simply $X = x^* {x^*}^{\top}$, then we trivially recover $x^*$ from the SDP solution. We thus ask: does the min cut condition of \cref{fact:infotheoretic,lem:cutsparsification} imply that $x^* {x^*}^{\top}$ is the unique optimal solution to the SDP? Namely, is the min cut condition sufficient for the SDP to certify that $x^*$ is the unique optimal assignment?

Unfortunately, it turns out that this is not the case, and we give a counterexample in \cref{sec:sparsification-example}. We thus require a stronger condition than the min cut one in order to obtain efficient algorithms. Nonetheless, an analogue of \cref{fact:infotheoretic} continues to hold, although now we require a stronger version that holds for all SDP solutions $X$, not just $x \in \Fits^n$. This stronger statement shows the SDP can \emph{certify} that $x^*$ is the unique optimal assignment if and only if a certain relative spectral approximation guarantee holds for the corrupted edges.

\begin{lemma}[SDP-certified uniqueness from relative spectral approximation]
\label{lem:sdpuniqueness}
    Let $G = (V,E_G)$ be an $n$-vertex connected graph, and let $H = (V, E_H)$ be a subgraph of $G$ where $E_H \subseteq E_G$.
    Let $L_G, L_H$ be the unnormalized Laplacians of $G$ and $H$.
    Consider a noisy planted 2-XOR instance $\phi$ on $G$ with planted assignment $x^* \in \Fits^n$ (\cref{def:semirandomxor}), and suppose $E_H$ is the set of corrupted edges.

    The SDP relaxation of $\phi$ satisfies
    \begin{equation*}
        \max_{X \succeq 0,\ X = X^{\top},\ \diag(X)=\Id} \phi(X) = \phi(x^*) = |E_G| - 2|E_H| \mcom
    \end{equation*}
    where $X = x^* x^{*\top}$ is the \emph{unique} optimum if and only if $G$ and $H$ satisfy
    \begin{align*}
        \Iprod{X, L_H} < \frac{1}{2} \Iprod{X, L_G},\quad
        \forall X\succeq 0,\ X = X^{\top},\ \diag(X) = \Id,\ X \neq \one \one^\top \mper
    \end{align*}
\end{lemma}
\begin{proof}
    Recall that each $e = \{i,j\}\in E$ corresponds to a constraint $x_i x_j = b_{e}$ where $b_{e} = x_i^* x_j^*$ if $e \in E_G \setminus E_H$ and $b_{e} = - x_i^* x_j^*$ if $e\in E_H$, meaning that $\phi(X) = \sum_{\{i,j\}\in G\setminus E} X_{ij} x_i^* x_j^* - \sum_{\{i,j\}\in E} X_{ij} x_i^* x_j^*$.
    Without loss of generality, we can assume that $x^* = \one$ and that $\phi(X) = \frac{1}{2} \iprod{X, A_G - 2A_H}$, where $A_G$, $A_H$ are the adjacency matrices of $G$ and $H$.

    Note that $L_G = D_G - A_G$ and $L_H = D_H - A_H$, and $\tr(D_G) = 2|E_G|$, $\tr(D_H) = 2|E_H|$.
    For any $X \succeq 0$ with $\diag(X) = \Id$,
    \begin{equation*}
        \Iprod{X, A_G - 2A_H} = \Iprod{X, (D_G - L_G)- 2(D_H - L_H)}
        = 2(|E_G| - 2|E_H|) + \Iprod{X, 2L_H - L_G} \mper
    \end{equation*}
    Suppose $\iprod{X, L_H} < \frac{1}{2} \iprod{X, L_G}$ for all $X \neq \one\one^\top$.
    Since $\iprod{\one\one^\top, L_G} = \iprod{\one\one^\top, L_H} = 0$, we have that the maximum of $\frac{1}{2}\iprod{X, A_G - 2A_H}$ is $|E_G| - 2|E_H|$ and $X = \one\one^\top$ is the unique maximum.

    For the other direction, suppose there is an $X \neq \one\one^\top$ such that $\iprod{X, L_H} \geq \frac{1}{2} \iprod{X, L_G}$.
    Then, $\phi(X) \geq |E_G| - 2|E_H| = \phi(\one\one^\top)$, meaning that $\one\one^\top$ is not the unique optimum.
\end{proof}
\parhead{Relative spectral approximation from uniform subsamples.}
We now come to a key technical observation. Suppose that $H$ is a \emph{spectral sparsifier} of $G$, so that $v^{\top} (\frac{1}{\eta} L_H) v$ is  $(1 \pm \delta) v^{\top} L_G v$ for any $v \in \R^n$. Then clearly $\ip{X, L_H} < \frac{1}{2}\ip{X, L_G}$ if $\eta < 1/2$ and $\delta = o(1)$, as we can write $X = \sum_{i = 1}^n \lambda_i v_i v_i^{\top}$, and
\begin{equation*}
\ip{X, L_H} = \sum_{i = 1}^n \lambda_i v_i^{\top} L_H v_i \leq \eta(1 + \delta) \sum_{i = 1}^n \lambda_i v_i^{\top} L_G v_i = \eta(1 + \delta) \cdot \ip{X, L_G} < \frac{1}{2} \ip{X, L_G} \mper
\end{equation*}
Furthermore, note that above we only required that $L_H \preceq \eta(1 + \delta) L_G$, i.e., we only use the upper part of the spectral approximation.

We are now ready to state the key relative spectral approximation lemma. We observe that when $H$ is a uniformly random subsample of $G$ and $G$ has a \emph{spectral gap} and minimum degree $\polylog(n)$, then with high probability $L_H \preceq \eta(1 + \delta) L_G$.
We note that, while we do not provide a formal proof, the same argument using the lower tail of Matrix Chernoff can also establish a lower bound on $L_H$, which proves that $H$ is indeed a spectral sparsifier of $G$.
\begin{lemma}[Relative spectral approximation from uniform subsamples]
\label{lem:spectral-sparsification}
    Let $\eta \in (0,1)$.
    Suppose $G = (V,E)$ is an $n$-vertex graph with minimum degree $d_{\min}$ (self-loops allowed) and spectral gap $\lambda_2(\widetilde{L}_G) = \lambda$ such that $d_{\min}\lambda > \frac{18}{\eta} \log n$, where  $\widetilde{L}_G := D_G^{-1/2} {L}_G D_G^{-1/2}$ is the \emph{normalized} Laplacian.
    Let $H$ be a subgraph of $G$ obtained by selecting each edge with probability $\eta$.
    Then, with probability at least $1-O(n^{-2})$,
    \begin{align*}
        L_H \preceq \eta(1+\delta) \cdot L_G
    \end{align*}
    for $\delta = \sqrt{\frac{18 \log n}{\eta d_{\min} \lambda}}$.
\end{lemma}
\begin{proof}
    First, note that $\one$ lies in the kernel of both $L_G$ and $L_H$, and because of the spectral gap of $G$, $\dim(\ker(L_G)) = 1$.
    Therefore, recalling that $L_G = D_G^{1/2} \widetilde{L}_G D_G^{1/2}$, it suffices to prove that
    \begin{align*}
        \Norm{ (\widetilde{L}_G^{\dagger})^{1/2} D_G^{-1/2} L_H D_G^{-1/2} (\widetilde{L}_G^{\dagger})^{1/2} }_2 \leq \eta(1+\delta) \mper
    \end{align*}
    Here $\widetilde{L}_G^{\dagger}$ is the pseudo-inverse of $\widetilde{L}_G$, and $\|\widetilde{L}_G^{\dagger}\|_2 \leq 1/\lambda$ because $G$ has spectral gap $\lambda$.
    We will write $X \coloneqq (\widetilde{L}_G^{\dagger})^{1/2} D_G^{-1/2} L_H D_G^{-1/2} (\widetilde{L}_G^{\dagger})^{1/2}$ for convenience.

    Note that $L_G = \sum_{e\in E} L_e$, where $L_e \succeq 0$ is the Laplacian of a single edge $e$ and $\|L_e\|_2 = 2$.
    Let $X_e = (\widetilde{L}_G^{\dagger})^{1/2} D_G^{-1/2} L_e D_G^{-1/2} (\widetilde{L}_G^{\dagger})^{1/2}$ if $e$ is chosen in $H$ and $0$ otherwise.
    Then, $X = \sum_{e\in E} X_e$ and $\|\E[X]\|_2 = \eta$.
    Moreover, each $X_e$ satisfies $X_e \succeq 0$ and $\|X_e\|_2 \leq \|\widetilde{L}_G^{\dagger}\|_2 \cdot \|D_G^{-1}\|_2 \cdot \|L_e\|_2 \leq \frac{2}{ d_{\min} \lambda}$.
    Thus, by Matrix Chernoff (\Cref{fact:matrix-chernoff}),
    \begin{align*}
        \Pr\Brac{\|X\|_2 \geq \eta(1+\delta)} \leq n \cdot \exp\Paren{-\frac{\delta^2 \eta}{3} \cdot \frac{d_{\min} \lambda}{2}} \leq O(n^{-2})
    \end{align*}
    as long as $\frac{18 \log n}{\eta d_{\min} \lambda} \leq \delta^2 \leq 1$.
\end{proof}
\parhead{Finishing the algorithm.}
By \cref{lem:sdpuniqueness,lem:spectral-sparsification}, we can thus recover $x^*$ exactly if the constraint graph $G$ of $\phi$ has a nontrivial spectral gap and minimum degree $d_{\min} \geq \polylog(n)$. To finish the implementation of \cref{strat:alg}, we thus need to explain how to algorithmically decompose any graph $G$ into subgraphs $G_1, \dots, G_T$, each with reasonable min degree and nontrivial spectral gap, while only discarding a $o(1)$-fraction of the edges in $G$. This is the well-studied task of expander decomposition, for which we appeal to known results \cite{KannanVV04, SpielmanT11, Wulff17, SaranurakW19}.

This completes the high-level description of the algorithm in the even $k$ case. Below, we summarize the steps of the final algorithm.
\begin{mdframed}
    \begin{algorithm}[Algorithm for $k$-XOR for even $k$]
    \label{alg:evenk}
    \mbox{}
      \begin{description}
      \item[Input:] $k$-XOR instance $\psi$ on $n$ variables with $m$ constraints and constraint hypergraph $\cH$.

      \item[Output:] Disjoint sets of constraints $\cA_1, \cA_2 \subseteq \cH$ such that $\abs{\cA_1} \leq o(m)$ and only depends on $\cH$, and $\cA_2 = (\cH \setminus \cA_1) \cap \Err_{\psi}$.

      \item[Operation:] \mbox{}
        \begin{enumerate}
            \item Construct the $2$-XOR instance $\phi$ with constraint graph $G$, as described in \cref{def:reductionto2xor}.
            \item Remove small-degree vertices and run expander decomposition on $G$ to produce expanders $G_1, \dots, G_T$. Set $\cA_1$ to be the set of discarded constraints of size $o(m)$.
            \item For each $i \in [T]$, solve the basic SDP on the subinstance $\phi_i$ defined by the constraints $G_i$. Let $\cA^{(i)}_2$ denote the set of constraints violated by the optimal local SDP solution.
            \item Output $\cA_1$ and $\cA_2 = \bigcup_{i = 1}^T \cA^{(i)}_2$.
        \end{enumerate}
      \end{description}
    \end{algorithm}
\end{mdframed}

\subsection{The case of odd \texorpdfstring{$k$}{k}}
\label{sec:oddk}
We are now ready to briefly explain the differences in the case when $k$ is odd. For the purposes of this overview, we will focus only on the case of $k = 3$. Recall that we are given a $3$-XOR instance $\psi$, specified by a $3$-uniform hypergraph $\cH \subseteq \binom{[n]}{3}$, as well as the right-hand sides $b_C \in \Fits$ for $C \in \cH$, where $b_C = x^*_C$ with probability $1 - \eta$ and $b_C = -x^*_C$ otherwise and $x^* \in \Fits^n$ is the planted assignment.

We now produce a $4$-XOR instance using the well-known ``Cauchy-Schwarz trick'' from CSP refutation \cite{CojaGL07}. The general idea is to, for any pair of clauses $(C, C')$ that intersect, add the ``derived constraint'' $x_C x_{C'} = b_C b_{C'}$ to the $4$-XOR instance. Notice that if, e.g., $C = \{u, i, j\}$ and $C' = \{u, i', j'\}$, then $x_u$ appears twice on the left-hand side, and thus the constraint is $x_i x_j x_{i'} x_{j'} = b_C b_{C'}$. Given this $4$-XOR, we produce a $2$-XOR following a similar strategy as in \cref{def:reductionto2xor}. The above description omits many technical details, which we handle in \cref{sec:decomposition,sec:bipartite-kXOR}; we remark here that these are the same issues that arise in the CSP refutation case, and we handle them using the techniques in \cite{GuruswamiKM22}.

We have thus produced a $2$-XOR instance $\phi$ that is noisy but not in the sense of \cref{def:semirandomxor}. Indeed, each edge $e$ in $\phi$ is ``labeled'' by a pair $(C,C')$ of constraints in $\psi$, and $e$ is noisy if and only if \emph{exactly} one of $(C,C')$ is, and so the noise is not independent across constraints. Nonetheless, we can still follow the general strategy as in \cref{alg:evenk}. The main technical challenge is to argue that the relative spectral approximation guarantee of \cref{lem:spectral-sparsification} holds even when the noise has the aforementioned correlations, and we do this in \cref{lem:spectral-sparsification-odd}. This allows us to recover, for most intersecting pairs $(C,C')$, the quantity $\xi(C) \xi(C')$, where $\xi(C) = -1$ if $C$ is corrupted, and is $1$ otherwise, i.e., $b_C = x^*_C \xi(C)$; we do not determine $\xi(C) \xi(C')$ if and only if the pair $(C,C')$ corresponds to an edge $e$ that was discarded during the expander decomposition.

However, we are not quite done, as we would like to recover $\xi(C)$ for most $C$, but we only know $\xi(C) \xi(C')$ for most intersecting pairs $(C,C')$. Let us proceed by assuming that we know $\xi(C) \xi(C')$ for all intersecting pairs $(C,C')$, and then we will explain how to do a similar decoding process when we only know most pairs. Let us fix a vertex $u$, and let $\cH_u$ denote the set of $C \in \cH$ containing $u$. Now, we know $\xi(C) \xi(C')$ for all $C, C' \in \cH_u$, and so by Gaussian elimination we can determine $\xi(C)$ for all $C \in \cH_u$ up to a global sign.
Now, we know that the vector $\{\xi(C)\}_{C \in \cH_u}$ should have roughly $\eta \abs{\cH_u}$ entries that are $-1$. So, choosing the global sign that results in fewer $-1$'s, we thus correctly determine $\xi(C)$ for all $C \in \cH_u$. We can then repeat this process for each choice of $u$ to decode $\xi(C)$ for all $C$.

Of course, we only actually know $\xi(C) \xi(C')$ for most intersecting pairs $(C,C')$. This implies that for most choices of $u$, the graph $G_u$ with vertices $\cH_u$ and edges $(C,C')$ if we know $\xi(C) \xi(C')$ is obtained from the complete graph on vertices $\cH_u$ and deleting some $o(1)$-fraction of edges. This implies that $G_u$ has a connected component of size $(1 - o(1))\abs{\cH_u}$, and again via Gaussian elimination and picking the proper global sign, we can determine $\xi(C)$ on this large connected component. By repeating this process for each choice of $u$, we thus recover $\xi(C)$ for most $u$.

\subsection{Organization}
The rest of the paper is organized as follows. In \cref{sec:prelims}, we introduce some notation, and recall the various concentration inequalities and facts that we will use in our proofs. In \cref{sec:reductiontoxor}, we prove \cref{mthm:main} from \cref{mthm:algxor} by reducing semirandom planted CSPs to noisy XOR. In \cref{sec:decomposition,sec:bipartite-kXOR}, we prove \cref{mthm:algxor}; \cref{sec:decomposition} handles the reduction from $k$-XOR to ``bipartite $k$-XOR'', and then \cref{sec:bipartite-kXOR} gives the algorithm for the bipartite $k$-XOR case.

\section{Preliminaries}
\label{sec:prelims}

\paragraph{Notation.}
Given a graph $G = (V, E)$ with $n$ vertices and $m$ edges (including self-loops\footnote{Each self-loop contributes $1$ to the degree of a vertex.}), we write $D_G \in \R^{n\times n}$ as the diagonal degree matrix, $A_G \in \R^{n\times n}$ as the adjacency matrix, and $L_G = D_G - A_G$ as the unnormalized Laplacian (note that the self-loops do not contribute to $L_G$).
Furthermore, we write $\widetilde{L}_G = D_G^{-1/2} L_G D_G^{-1/2}$ to be the \emph{normalized} Laplacian, and denote its eigenvalues as $0 = \lambda_1(\widetilde{L}_G) \leq \lambda_2(\widetilde{L}_G) \leq \cdots \leq \lambda_n(\widetilde{L}_G) \leq 2$.

For any subset $S \subseteq V$, we denote $G[S]$ as the subgraph of $G$ induced by $S$, and $G\{S\}$ as the induced subgraph $G[S]$ but with self-loops added so that any vertex in $S$ has the same degree as its degree in $G$.

\begin{definition}[Uniform hypergraphs]
A $k$-uniform hypergraph $\cH$ on $n$ vertices is a collection $\cH$ of subsets of $[n]$ of size exactly $k$. For a set $\setQ \subseteq [n]$, we define $\deg(\setQ) := \abs{\{C \in \cH : \setQ \subseteq C\}}$.
\end{definition}

\subsection{Concentration inequalities}

\begin{fact}[Chernoff bound]  \label{fact:chernoff}
    Let $X_1,\dots,X_n$ be independent random variables taking values in $\zo$.
    Let $X = \sum_{i=1}^n X_i$ and $\mu = \E[X]$.
    Then, for any $\delta \in [0,1]$,
    \begin{align*}
        \Pr\Brac{|X-\mu| \geq \delta \mu} \leq 2 e^{-\delta^2\mu /3} \mper
    \end{align*}
\end{fact}

\begin{fact}[{Matrix Chernoff~\cite[Theorem 5.1.1]{Tropp15}}] \label{fact:matrix-chernoff}
    Let $X_1,\dots, X_n\in \R^{d\times d}$ be independent, random, symmetric matrices such that $X_i \succeq 0$ and $\lambda_{\max}(X_i) \leq R$ almost surely.
    Let $X = \sum_{i=1}^n X_i$ and $\mu = \lambda_{\max}(\E[X])$.
    Then, for any $\delta \in [0,1]$,
    \begin{align*}
        \Pr\Brac{\lambda_{\max}(X) \geq (1+\delta) \mu} \leq d \cdot \exp\Paren{-\frac{\delta^2\mu}{3R}} \mper
    \end{align*}
\end{fact}

\subsection{Graph pruning and expander decomposition}

It is a standard result that given a graph with $m$ edges and average degree $d$, one can delete vertices such that the resulting graph has minimum degree $\eps d$ and at least $(1-2\eps)m$ edges. We include a short proof for completeness.
\begin{lemma}[Graph pruning] \label{lem:graph-pruning}
    Let $G$ be an $n$-vertex graph with average degree $d$ and $m = \frac{nd}{2}$ edges, and let $\eps \in (0,1/2)$.
    There is an algorithm that deletes vertices of $G$ such that the resulting graph has minimum degree $\eps d$ and at least $(1-2\eps)m$ edges.
\end{lemma}
\begin{proof}
    The algorithm is simple: repeatedly remove any vertex with degree $<\eps d$.
    First, we show by induction that each deletion cannot decrease the average degree.
    Suppose there are $n' \leq n$ vertices left and average degree $d' \geq d$. Then, after deleting a vertex $u$ with degree $d_u < \eps d$, the average degree becomes $\frac{n'd'-2d_u}{n'-1} > \frac{n' d - 2\eps d}{n'-1} = d \cdot \frac{n'-2\eps}{n'-1}$.
    Thus, for $\eps < 1/2$, the average degree is always at least $d$.
    Furthermore, since the algorithm can delete at most $n$ vertices, it can delete at most $\eps d n = 2\eps m$ edges.
\end{proof}

We will also need an algorithm that partitions a graph into expanding clusters such that total number of edges across different clusters is small.
Expander decomposition has been developed in a long line of work \cite{KannanVV04,SpielmanT11,Wulff17,SaranurakW19} and has a wide range of applications. For our algorithm, we only require a very simple expander decomposition that recursively applies Cheeger's inequality.
\begin{fact}[Expander decomposition]
\label{fact:expander-decomp}
    Given a (multi)graph $G = (V,E)$ with $m$ edges and a parameter $\eps \in (0,1)$, there is a polynomial-time algorithm that finds a partition of $V$ into $V_1,\dots, V_T$ such that $\lambda_2(\widetilde{L}_{G\{V_i\}}) \geq \Omega(\eps^2/\log^2 m)$ for each $i\in[T]$ and the number of edges across partitions is at most $\eps m$.
 \end{fact}
 \begin{proof}
    Fix $\lambda = c\eps^2/\log^2 m$ for some constant $c$ to be chosen later.
    The algorithm is very simple. Given a graph $G = (V,E)$ (with potentially parallel edges and self-loops), if $\lambda_2(\wt{L}_G) < \lambda$, then by Cheeger's inequality we can efficiently find a subset $S\subseteq V$ with $\vol(S) \leq \vol(\ol{S})$ such that $\frac{|E(S, \ol{S})|}{\vol(S)} < \sqrt{2\lambda}$.
    Here $\vol(S) \defeq \sum_{v\in S} \deg(v)$.
    Then, we cut along $S$, add self-loops to the induced subgraphs $G[S]$ and $G[\ol{S}]$ so that the vertex degrees remain the same (each self-loop contributes $1$ to the degree). This produces two graphs $G\{S\}$ and $G\{\ol{S}\}$, and we recurse on each.
    By construction, in the end we will have partitions $V_1,\dots,V_T$ where either $V_i$ is either a single vertex or satisfies $\lambda_2(\wt{L}_{G\{V_i\}}) \geq \lambda$.

    We now bound the number of edges cut via a charging argument.
    Consider the ``half-edges'' in the graph, where each edge $(u,v)$ contributes one half-edge to $u$ and one to $v$, and each self-loop counts as one half-edge.
    Then, $\vol(S)$ equals the number of half-edges attached to $S$.
    Now, imagine we have a counter for each half-edge, and every time we cut along $S$ we add $\sqrt{2\lambda}$ to each half-edge attached to $S$ (the smaller side).
    Since $E(S,\ol{S}) < \sqrt{2\lambda}\cdot \vol(S)$, it follows that the number of edges cut is at most the total sum of the counters.
    On the other hand, each half-edge can appear on the smaller side of the cut at most $\log_2 2m$ times, as each time the half-edge is on the smaller side of the cut, $\vol(S)$ decreases by at least a factor of $2$, and $\vol([n]) = 2m$.
    So, the total sum must be $\leq \sqrt{2\lambda} \cdot 2m \log_2 2m \leq \eps m$ for a small enough constant $c$.
 \end{proof}
\section{From Planted CSPs to Noisy XOR}
\label{sec:reductiontoxor}
In this section, we show how to use \cref{mthm:algxor} to prove \cref{mthm:main}.
Before we delve into the formal proof, we will first explain the reduction given in \cite{FeldmanPV15}. We begin with some definitions.

\parhead{Setup.} Let $\Psi$ be sampled from $\Psi(\vec{\cH}, x^*, Q)$, where $x^* \in \Fits^n$, $\vec{\cH} \subseteq [n]^k$, and $Q$ is a planting distribution for the predicate $P$. Let $Q(y) = \sum_{S \subseteq [k]} \hat{Q}(S) \prod_{i \in S} y_i$ be the Fourier decomposition of $Q$, where $\hat{Q}(S) = \frac{1}{2^k} \sum_{y \in \Fits^k} Q(y) \prod_{i \in S} y_i \in [-2^{-k}, 2^{-k}]$.
Recall (\cref{def:semirandomcsp}) that $\Psi$ is specified by a collection $\vec{\cH} \subseteq [n]^k$ of scopes, along with a vector $\ell(\vec{C}) \in \Fits^k$ for each $\vec{C} \in \vec{\cH}$ of literal negations.

\begin{definition}
\label{def:xorsubinstances}
 Let $S \subseteq [k]$ be nonempty. Let $\psi^{(S,+)}$ be the
 $\abs{S}$-XOR instance obtained by, for each constraint $\vec{C}$ in $\Psi$, adding the constraint $\prod_{i \in S} x_{\vec{C}_i} = \prod_{i \in S} \ell(\vec{C})_i$. Similarly, let $\psi^{(S,-)}$ have constraints $\prod_{i \in S} x_{\vec{C}_i} = -\prod_{i \in S} \ell(\vec{C})_i$.
\end{definition}
We make use of the following simple claim.
\begin{claim}
\label{claim:noisysubinstance}
For each nonempty $S \subseteq [k]$, $\psi^{(S, +)}$ is a noisy $\abs{S}$-XOR instance (\cref{def:semirandomxor}) with planted assignment $x^*$ and noise $\eta = \frac{1}{2}(1 - 2^k \hat{Q}(S))$. Similarly, $\psi^{(S, -)}$ is a noisy $\abs{S}$-XOR instance with planted assignment $x^*$ and noise $\eta = \frac{1}{2}(1 + 2^k \hat{Q}(S))$.
\end{claim}
\begin{proof}
For each $\vec{C}$, the literal negation $\ell(\vec{C})$ is sampled such that $\Pr[\ell(\vec{C}) = \ell] = Q(\ell \odot x^*_{\vec{C}})$, where $\odot$ denotes the element-wise product. This is equivalent to sampling $y \gets Q$ and setting $\ell(\vec{C}) = y \odot x^*_{\vec{C}}$.
It thus follows that the probability that the constraint $\vec{C}$ produces a corrupted constraint in $\psi^{(S, +)}$ is
\begin{equation*}
\Pr_{y \gets Q} \Brac{\prod_{i \in S} y_i = -1} = \frac{1}{2} \Paren{1 - \E_{y \gets Q}\Brac{\prod_{i \in S} y_i} } = \frac{1}{2} (1 - 2^k \hat{Q}(S)) \enspace,
\end{equation*}
and is independent for each $\vec{C}$. A similar calculation handles the case of $\psi^{(S, -)}$.
\end{proof}
With the above observations in hand, we can now easily describe the reduction in \cite{FeldmanPV15}. First, their reduction requires the algorithm to have a description of the distribution $Q$. Given $Q$, the algorithm then finds the smallest $S$ such that $\hat{Q}(S)$ is nonzero. Since they know the exact value of $\hat{Q}(S)$, they can determine its sign correctly. Suppose that $\hat{Q}(S) > 0$ (the other case is similar). Then, by solving the $\abs{S}$-XOR instance $\psi^{(S, +)}$, they recover the planted assignment of $\psi^{(S, +)}$ \emph{exactly}.\footnote{Here, they also treat $\abs{\hat{Q}(S)}$ as constant, as if $\abs{\hat{Q}(S)} \ll 1/n$, say, then their algorithm would not succeed in recovering the planted assignment on the XOR instance.}
But this planted assignment is precisely $x^*$, and so they have also succeeded in recovering the planted assignment of $\psi$.

The aforementioned reduction clearly does not generalize to the semirandom setting, as in general the subinstances $\psi^{(S, \pm)}$ will not uniquely determine $x^*$. Furthermore, their reduction additionally requires knowing $Q$, and while it is not too unreasonable to assume this for random planted CSPs (as it is perhaps natural for the algorithm to know the distribution), in the semirandom setting this assumption is a bit strange because we want to view semirandom CSPs as ``moving towards'' worst case ones.

We now prove \cref{mthm:main} from \cref{mthm:algxor}.
\begin{proof}[Proof of \cref{mthm:main} from \cref{mthm:algxor}]
We will present the proof in three steps. 
First, like \cite{FeldmanPV15}, we will assume that the algorithm is given a description of $Q$ and we will assume that each $\abs{\hat{Q}(S)}$ is either $0$ or at least $2^{-k} \eps > 0$.\footnote{This assumption is implicit in \cite{FeldmanPV15}; see the previous footnote.} Then, we will remove this assumption provided that $Q(y) > 2 \eps$ for all $y$ with $Q(y) > 0$, i.e., the every $y$ in the support of $Q$ has some minimum probability. Finally, we will remove the last assumption.

\parhead{Step 1: the proof when we are given $Q$.}
For each $S$ where $\hat{Q}(S) \ne 0$, we construct the instance $\psi^{(S,+)}$ (if $\hat{Q}(S) > 0$) or $\psi^{(S,-)}$ (if $\hat{Q}(S) < 0$). We then apply\footnote{Note that \cref{mthm:algxor} only applies when $\abs{S} \geq 2$. When $\abs{S} = 1$, there is a trivial algorithm; see \cref{sec:1xor} for details.} \cref{mthm:algxor} to each such instance.
Note that by \cref{claim:noisysubinstance}, the instance has noise $\eta = \frac{1}{2}(1 - 2^k |\hat{Q}(S)|) \leq \frac{1}{2}(1 - \eps)$ (because we picked the correct sign when choosing between $\psi^{(S,+)}$ and $\psi^{(S,-)}$, and we assume $\abs{\hat{Q}(S)} \geq 2^{-k} \eps$).
Then, since $m \geq c^k n^{k/2} \cdot \frac{\log^3 n}{\eps^9}$ and $\abs{S} \leq k$, by applying \cref{mthm:algxor} with noise $\eta$ and parameter $\eps' \defeq 2^{-k} \eps$, we obtain sets $\vec{\cH}^{(S, 1)}$ (the discarded set) and $\vec{\cH}^{(S,2)}$ (the corrupted constraints) where $\abs{\vec{\cH}^{(S, 1)}} \leq \eps' m$ and $\vec{\cH}^{(S,2)} = (\vec{\cH} \setminus \vec{\cH}^{(S, 1)})\cap \Err_{\psi^{(S)}}$.
Hence, for every constraint $\vec{C} \in \vec{\cH} \setminus \vec{\cH}^{(S,1)}$, it follows that we have learned $\prod_{i \in S} x^*_{\vec{C}_i}$, where $x^*$ is the planted assignment for $\Psi$. By setting $\vec{\cH'} \defeq \vec{\cH} \setminus \cup_{S : \hat{Q}(S) \ne 0} \vec{\cH}^{(S, 1)}$, it follows that we know $\prod_{i \in S} x^*_{\vec{C}_i}$ for all $ \vec{C} \in \vec{\cH'}$ and $S$ with $\hat{Q}(S) \ne 0$, where $\abs{\vec{\cH'}} \geq (1 - 2^k \eps')m = (1-\eps)m$.

We now solve the system of linear equations given by $\prod_{i \in S} x^*_{\vec{C}_i}$ for all $ \vec{C} \in \vec{\cH'}$ and $S$ with $\hat{Q}(S) \ne 0$ to obtain some assignment $x \in \Fits^n$. As $x^*$ is a valid solution to these equations, such an $x$ exists, although it may not be $x^*$.

The final step is to argue that for every $\vec{C} \in \vec{\cH'}$, $x$ satisfies the constraint $\vec{C}$, namely that $P(\ell(\vec{C})_1 x_{\vec{C}_1}, \ell(\vec{C})_2 x_{\vec{C}_2}, \ldots, \ell(\vec{C})_k x_{\vec{C}_k}) = 1$. Indeed, if this is true then we are done, as $x$ satisfies at least $(1 - \eps) m$ constraints in $\Psi$, and so we have obtained the desired assignment.

Let $\vec{C} \in \vec{\cH'}$. We know that for every $S$ with $\hat{Q}(S) \ne 0$, we have that $\prod_{i \in S} x_{\vec{C}_i} = \prod_{i \in S} x^*_{\vec{C}_i}$. Hence, it follows that 
\begin{equation*}
Q(\ell(\vec{C}) \odot x) = \sum_{S \subseteq [k]} \hat{Q}(S) \prod_{i \in S} \ell(\vec{C})_i x_{\vec{C}_i} = \sum_{S \subseteq [k]} \hat{Q}(S) \prod_{i \in S} \ell(\vec{C})_i x^*_{\vec{C}_i} = Q(\ell(\vec{C}) \odot x^*) > 0 \mcom
\end{equation*}
where the last inequality is because $\ell(\vec{C})$ was sampled from the distribution $Q(\ell(\vec{C}) \odot x^*)$, and so it must be sampled with nonzero probability. As $Q$ is supported only on satisfying assignments to the predicate $P$, it thus follows that $\ell(\vec{C}) \odot x^*$ must also satisfy $P$.

\parhead{Step 2: removing the dependence on $Q$ assuming a lower bound on $Q(y)$.} First, we observe that because $k$ is constant, we can, for each $S$, guess a symbol $\{0, +, -\}$, where $0$ denotes, informally, the belief that $\abs{\hat{Q}(S)} < 2^{-k}\eps$, $+$ denotes that $\hat{Q}(S) \geq 2^{-k}\eps$, and $-$ denotes that $\hat{Q}(S) \leq -2^{-k}\eps$. For each of the $3^{2^k}$ choices of guesses, i.e., functions $f \colon \{S \subseteq [k]\} \to \{0, +, -\}$, we run algorithm mentioned in the previous step. Namely, for each $S$: \begin{inparaenum}[(1)]
\item if $f(S) = 0$, then we ignore $S$,
\item if $f(S) = +$, then we run \cref{mthm:algxor} on $\psi^{(S,+)}$ to obtain $\vec{\cH}^{(S, 1)}$ and $\vec{\cH}^{(S,2)}$, and 
\item if $f(S) = -$, then we run \cref{mthm:algxor} on $\psi^{(S,+)}$ to obtain $\vec{\cH}^{(S, 1)}$ and $\vec{\cH}^{(S,2)}$.
\end{inparaenum}
As before, we solve the system of linear equations to obtain some assignment $x^{(f)} \in \Fits^n$.
By enumerating over all possible choices of $f$, we obtain a list of at most $3^{2^k} = O(1)$ assignments. We then try all of them and output the best one.

It thus remains to show that at least one of the assignments in the list has high value. As one may expect, this will be the assignment $x^{(f^*)}$, where $f^*$ is the correct label function. Indeed, when $f = f^*$, then we are precisely running the algorithm in Step 1, and as observed, after solving the linear system of equations we obtain an assignment $x \defeq x^{(f^*)}$ with the following property. For every $\vec{C} \in \vec{\cH'}$ and every $S$ with $\abs{\hat{Q}(S)} \geq 2^{-k}\eps$, we have that $\prod_{i \in S} x_{\vec{C}_i} = \prod_{i \in S} x^*_{\vec{C}_i}$, where $\vec{\cH'} \subseteq \vec{\cH}$ has size $\geq (1 - \eps)m$.

Finally, we show that for every $\vec{C} \in \vec{\cH'}$, $x$ satisfies the constraint $\vec{C}$. Namely, we have $P(\ell(\vec{C})_1 x_{\vec{C}_1}, \ell(\vec{C})_2 x_{\vec{C}_2}, \ldots, \ell(\vec{C})_k x_{\vec{C}_k}) = 1$. Let $\vec{C} \in \vec{\cH'}$. We know that for every $S$ with $\abs{\hat{Q}(S)} \geq 2^{-k}\eps$, we have that $\prod_{i \in S} x_{\vec{C}_i} = \prod_{i \in S} x^*_{\vec{C}_i}$. Hence, it follows that 
\begin{flalign*}
\Abs{Q(\ell(\vec{C}) \odot x) - Q(\ell(\vec{C}) \odot x^*)} &= \Abs{\sum_{S \subseteq [k]} \hat{Q}(S) \prod_{i \in S} \ell(\vec{C})_i x_{\vec{C}_i} - \sum_{S \subseteq [k]} \hat{Q}(S) \prod_{i \in S} \ell(\vec{C})_i x^*_{\vec{C}_i} } \\
&= \Abs{\sum_{S \subseteq [k] : \abs{\hat{Q}(S)} < 2^{-k}\eps} \hat{Q}(S) \Paren{\prod_{i \in S} \ell(\vec{C})_i x_{\vec{C}_i} - \prod_{i \in S} \ell(\vec{C})_i  x^*_{\vec{C}_i} }} \leq 2^k \cdot 2^{-k+1} \eps \mper
\end{flalign*}
Now, if we assume that $Q(y) > 2\eps$ for every $y \in \Fits^k$ with $Q(y) > 0$, then it follows that $Q(\ell(\vec{C}) \odot x) > 0$, and so $x$ satisfies the constraint $P(\ell(\vec{C})_1 x_{\vec{C}_1}, \ell(\vec{C})_2 x_{\vec{C}_2}, \ldots, \ell(\vec{C})_k x_{\vec{C}_k}) = 1$.

\parhead{Step 3: removing the lower bound on $Q(y)$.} In Step 2, we assumed that $Q(y) > 2\eps$ for all $y \in \Fits^k$ with $Q(y) > 0$. However, we only used this fact in the final step, when we argue that $Q(\ell(\vec{C}) \odot x) > 0$ by observing that $Q(\ell(\vec{C}) \odot x) \geq Q(\ell(\vec{C}) \odot x^*) - 2 \eps > 0$. To remove the assumption, we will show that for at most $2^{k+2} \eps$ constraints $\vec{C} \in \vec{\cH}$, it holds that $Q(\ell(\vec{C}) \odot x^*) \leq 2\eps$. This then implies that $x$ satisfies at least $(1 - \eps - 2^{k+2} \eps)m = (1 - O(\eps))m$ constraints, which finishes the proof.

Let $\cS$ denote the set of $\vec{C} \in \vec{\cH}$ where $Q(\ell(\vec{C}) \odot x^*) \leq 2 \eps$. Observe that the probability, over the choice of $\ell(\vec{C})$, that $\vec{C} \in \cS$ is at most $2^k \cdot 2 \eps = 2^{k+1} \eps$, and moreover this is independent for each $\vec{C} \in \vec{\cH}$. Thus, by a Chernoff bound, it follows that with probability $\geq 1 - \exp(-O(\eps m)) \geq 1 - 1/\poly(n)$, it holds that $\abs{\cS} \leq 2 \cdot 2^{k+1} \eps$, and so we are done.
\end{proof}

\begin{remark}[Tolerating fewer constraints for structured $Q$'s]
\label{rem:distcomplexity}
We have shown that the above algorithm succeeds in finding an assignment $x$ that satisfies at least $(1 - O(\eps))m$ constraints when $m \geq n^{k/2} \cdot \poly(\log n, 1/\eps)$. However, if the distribution $Q$ has $\abs{\hat{Q}(S)} < 2^{-k} \eps$ for all $S$ with $\abs{S} > r$, then we only need $n^{r/2} \cdot \poly(\log n, 1/\eps)$ constraints. (If $r = 0$, then for small enough constant $\eps$, $Q$ will be supported on all of $\Fits^k$, and so any assignment satisfies all constraints. If $r = 1$, we require $O(n \cdot \frac{\log n}{\eps})$ constraints; see \cref{lem:alg1xor}.)
Indeed, this follows because for such $Q$, the true label function $f^*$ will have $f^*(S) = 0$ for any $S$ with $\abs{S} > r$. Hence, for this choice of $f^*$, we only call \cref{mthm:algxor} on noisy $t$-XOR instances for $t \leq r$, and so we have enough constraints. It therefore follows that the assignment $x^{(f^*)}$ that we obtain for the label function $f^*$ will be, with high probability an assignment that satisfies at least $(1 - O(\eps))m$ constraints.

An example where this gives an improvement is the well-studied NAE-$3$-SAT (not-all-equal-3SAT) predicate~\cite{AnderssonE98, AchlioptasCIM01, DingSS14}.
Suppose $Q$ is the uniform distribution over satisfying assignments to NAE-$3$-SAT: $Q(x_1,x_2,x_3) = \frac{1}{6}\cdot \frac{1}{4}(3 - x_1x_2 - x_2 x_3 - x_1 x_3)$.
Then, we only need $m \geq \tilde{O}(n)$ constraints, even though it is a $3$-CSP ($k=3$).
\end{remark}

\section{From \texorpdfstring{$k$}{k}-XOR to Spread Bipartite \texorpdfstring{$k$}{k}-XOR}
\label{sec:decomposition}
In this section, we begin the proof of \cref{mthm:algxor}.
See \cref{def:semirandomxor} for a reminder of our semirandom planted $k$-XOR model $\psi(\cH, x^*, \eta)$ given a $k$-uniform hypergraph $\cH$, assignment $x^* \in \Fits^n$, and noise parameter $\eta \in (0,1/2)$. Recall also that $\Err_{\psi}$ denotes the set of corrupted hyperedges.

We think of $\Alg_1(\cH)$ as the small set of edges that we discard (or give up on), and this will only depend on the hypergraph $\cH$. For the rest of the graph, the algorithm will correctly identify which edges are corrupted.

Our proof of \cref{mthm:algxor} goes via a reduction to \emph{spread bipartite $t$-XOR} instances for $t = 2, \dots, k$, which are $t$-XOR instances with some additional desired structure. Such instances were introduced in \cite{GuruswamiKM22} to study the refutation of semirandom $k$-XOR instances. The reduction here is nearly identical to the corresponding reduction in \cite[Section 4]{GuruswamiKM22}.

\begin{definition}[Spread bipartite $k$-XOR] \label{def:spread-bipartite-kXOR}
A $p$-bipartite $k$-XOR instance $\psi$ on $n$ variables with $m$ constraints is defined by a collection of $(k-1)$-uniform hypergraphs $\cH = \{\cH_u\}_{u \in [p]}$ on the vertex set $[n]$, as well as ``right-hand sides'' $b_{u, C}$ for each $u \in [p]$ and $C \in \cH_u$. There are two sets of variables of $\psi$: the ``normal'' variables $x_1, \dots, x_n$, and the ``special'' variables $y_1, \dots, y_p$. The constraints of $\psi$ are $y_u \prod_{i \in C} x_i = b_{u, C}$ for each $u \in [p]$, $C \in \cH_u$. 

We furthermore say that $\psi$ is \emph{$\tau$-spread} if it has the following additional properties:
\begin{enumerate}[(1)]
\item $\abs{\cH_u} = \frac{m}{p} \geq 2\floor{\frac{1}{2\tau^2}}$ and $\frac{m}{p}$ is even for each $u \in [p]$,
\item For each $u \in [p]$ and set $Q \subseteq [n]$, $\deg_u(Q) \leq \frac{1}{\tau^2}\max(1, n^{\frac{k}{2} - 1 - \abs{Q}})$.
\end{enumerate}

Analogously to \cref{def:semirandomxor}, we call $\psi$ a \emph{semirandom planted} instance with planted assignment $(x^*, y^*)$ and noise parameter $\eta$ if the right-hand sides $b_{u, C}$ are generated by setting $b_{u, C} = y^*_u \prod_{i \in C} x^*_i$ with probability $1 - \eta$ and $b_{u, C} = -y^*_u \prod_{i \in C} x^*_i$ otherwise, independently for each choice of $u, C$. For a choice of $x^*, y^*$, $\cH = \{\cH_u\}_{u \in [p]}$, and $\eta$, we call this distribution $\psi(\{\cH_u\}_{u \in [p]}, x^*, y^*, \eta)$. As before, if an edge $(u,C)$ has $b_{u, C} = -y^*_u \prod_{i \in C} x^*_i$, we call $(u,C)$ a \emph{corrupted} hyperedge, and we denote the set of corrupted hyperedges in $\psi$ by $\Err_{\psi}$.
\end{definition}

The main technical result of the paper is the following lemma, which gives an algorithm to find the noisy constraints in a semirandom planted $\tau$-spread bipartite $k$-XOR instance.
\begin{lemma}[Algorithm for $\tau$-spread bipartite $k$-XOR]
\label{lem:algbipartitexor}
Let $k\geq 2$, $n, p \in \N$, $\eps\in (0,1)$, $\eta \in [0, 1/2)$, and let $\gamma \coloneqq 1-2\eta > 0$.
Let $\tau \leq \frac{c\gamma}{\sqrt{k\log n}}$, and let $m \geq C n^{\frac{k-1}{2}} \sqrt{p} \cdot \BipartiteBeta$ for some universal constants $c, C$.
There is a polynomial-time algorithm $\Alg$ that takes as input an $\tau$-spread $p$-bipartite $k$-XOR instance $\psi$ with constraint hypergraph $\cH = \{\cH_u\}_{u \in [p]}$ and outputs two disjoint sets $\Alg_1(\cH), \Alg_2(\psi) \subseteq \cH$ with the following guarantee:
\begin{inparaenum}[(1)] \item for any instance $\psi$ with $m$ constraints, $\abs{\Alg_1(\cH)} \leq \eps m$ and $\Alg_1(\cH)$ only depends on $\cH$, and \item for any $x^* \in \Fits^n, y^* \in \Fits^p$ and any $\cH = \{\cH_u\}_{u \in [p]}$ with $\abs{\cH} \defeq \sum_{u \in [p]} \abs{\cH_u} \geq m$, with probability $1-\frac{1}{\poly(n)}$ over $\psi \gets \psi(\{\cH_u\}_{u \in [p]}, x^*, y^*, \eta)$, it holds that $\Alg_2(\psi) = \Err_{\psi} \cap (\cH \setminus \Alg_1(\cH))$. \end{inparaenum}
\end{lemma}

Note that as $\eta \to \frac{1}{2}$, $\gamma = 1-2\eta \to 0$ and $\tau \to 0$, which blows up $m$. This is the expected behavior since when $\eta = \frac{1}{2}$, it is impossible to recover the planted assignment since the signs of the constraints are uniformly random.

\subsection{Proof of \texorpdfstring{\cref{mthm:algxor}}{Theorem~\ref{mthm:algxor}} from \texorpdfstring{\cref{lem:algbipartitexor}}{Lemma~\ref{lem:algbipartitexor}}}

With \cref{lem:algbipartitexor}, we can finish the proof of \cref{mthm:algxor}.
The high-level idea of this proof is very simple. First, we decompose the $k$-XOR instance $\psi$ into subinstances $\psi^{(t)}$ for each $t = 2, \dots, k$, using a hypergraph decomposition algorithm very similar to the one used in \cite{GuruswamiKM22,HsiehKM23}. The algorithm and its guarantees are shown in \cref{sec:hypergraph-decomp}.
Then, we run the algorithm in \cref{lem:algbipartitexor} to identify a set of corrupted constraints and a small set of discarded constraints within each subinstance $\psi^{(t)}$. We then take the union of these outputs to be the final output of the algorithm.

\begin{proof}[Proof of \cref{mthm:algxor}]
We begin with the decomposition of $\psi$ into $\psi^{(2)}, \dots, \psi^{(k)}$ along with a set of ``discarded'' hyperedges $\cH^{(1)}$, which is done using \cref{alg:decomp} with spread parameter $\tau \coloneqq \frac{c(1-2\eta)}{\sqrt{k\log n}}$ where $c$ is the constant in \cref{lem:algbipartitexor}.
For each $t = 2, \dots, k$, $\psi^{(t)}$ is a semirandom (with noise $\eta$) planted $\tau$-spread $p^{(t)}$-bipartite $t$-XOR instance specified by $(t-1)$-uniform hypergraphs $\{\cH^{(t)}_u\}_{u \in [p^{(t)}]}$.

Let $m^{(t)} \coloneqq \sum_{u\in [p^{(t)}]} \abs{\cH_u^{(t)}}$.
\cref{alg:decomp} has the following guarantees:
\begin{enumerate}[(1)]
\item The runtime is $n^{O(k)}$,
\item For each $t\in \{2,\dots,k\}$ and $u\in[p^{(t)}]$, $\abs{\cH_u^{(t)}} = \frac{m^{(t)}}{p^{(t)}} = 2 \floor{\frac{1}{2\tau^2} \max(1, n^{ t - \frac{k}{2} - 1})}$; in particular, $|\cH_u^{(t)}|$ is even and is at least $2\floor{\frac{1}{2\tau^2}}$,
\item For each $t = 2, \dots, k$, the instance $\psi^{(t)}$ is $\tau$-spread,
\item The number of ``discarded'' hyperedges is $m^{(1)} \defeq \abs{\cH^{(1)}} \leq \frac{1}{k \tau^2} n^{\frac{k}{2}}$,

\item For $t\in \{2,\dots,k\}$, each $C\in \cH_u^{(t)}$ is obtained by removing $k-(t-1)$ vertices from an edge in the original hypergraph $\cH$. Thus, there is a one-to-one map $\decomp \colon \cH \to \cH^{(1)} \cup \bigcup_{t = 2}^k \{\cH^{(t)}_u\}_{u \in [p^{(t)}]}$, such that an edge $C \in \cH$ is corrupted if and only if the edge $\decomp(C)$ is corrupted in the instance $\psi^{(t)}$ that it lies in.
\end{enumerate}

For convenience, we denote $\gamma \defeq 1-2\eta$ and $\beta \defeq 4C \cdot \BipartiteBeta = \frac{4C}{c}\cdot \frac{k^2\log^2 n}{\gamma^3 \eps^{3/2}}$ where $C, c$ are the constants in \cref{lem:algbipartitexor}.
The algorithm in \cref{mthm:algxor} works as follows. First, it runs \cref{alg:decomp} to produce the instances $\psi^{(2)}, \dots, \psi^{(k)}$. Then, for each $t = 2, \dots, k$, if $m^{(t)} \geq n^{\frac{t-1}{2}} \sqrt{p^{(t)}} \cdot \beta$, we run \cref{lem:algbipartitexor} on $\psi^{(t)}$ and obtain, with probability $1 - 1/\poly(n)$, a set $A_1^{(t)}$ where $\abs{A_1^{(t)}} \leq \frac{\eps}{2} m^{(t)}$ and $A_2^{(t)} = \Err_{\psi^{(t)}} \setminus A_1^{(t)}$. Otherwise, if $m^{(t)} < n^{\frac{t-1}{2}} \sqrt{p^{(t)}} \cdot \beta$, we set $A_1^{(t)} = \cH^{(t)}$ and $A_2^{(t)} = \emptyset$. Finally, we output $\cA_1 \defeq \cH^{(1)} \cup \bigcup_{t = 2}^k {\decomp}^{-1}(A_1^{(t)})$ and $\cA_2 \defeq \bigcup_{t=2}^k {\decomp}^{-1}(A_2^{(t)})$, where $\decomp$ is the mapping in property (5) of \cref{alg:decomp}.

Note that $m^{(t)} = p^{(t)} |\cH_u^{(t)}| \geq p^{(t)} \cdot \frac{1}{2\tau^2} n^{t-\frac{k}{2}-1}$, which means $p^{(t)} \leq 2\tau^2 n^{\frac{k}{2}-t+1} m^{(t)}$, and since $\sum_{t} \sqrt{m^{(t)}} \leq \sqrt{k \sum_t m^{(t)}} \leq \sqrt{km}$ by Cauchy-Schwarz, we have
\begin{flalign*}
    \sum_{t=2}^k n^{\frac{t-1}{2}} \sqrt{p^{(t)}} \cdot \beta
    \leq O(\tau) \cdot n^{\frac{k}{4}} \sqrt{km} \cdot \beta
    \leq o(\eps) m
\end{flalign*}
as long as $m \gg n^{\frac{k}{2}} \cdot k\tau^2 \beta^2/\eps^2$.
Moreover, $m^{(1)} \leq \frac{1}{k\tau^2} n^{\frac{k}{2}} = \frac{\log n}{c^2 \gamma^2} n^{\frac{k}{2}} \leq o(\eps) m$.
One can verify, by plugging in $\beta$, that the lower bound on $m$ in \cref{mthm:algxor} suffices.

By union bound over $t$, it thus follows that
\begin{flalign*}
    \abs{\cA_1} \leq m^{(1)} + \sum_{t=2}^k \frac{\eps}{2} m^{(t)} + \sum_{t=2}^k n^{\frac{t-1}{2}} \sqrt{p^{(t)}} \beta
    \leq \eps m \mcom
\end{flalign*}
and $\cA_2 = \Err_{\psi} \setminus \cA_1$. Moreover, by \cref{lem:algbipartitexor}, $\cA_1$ only depends on the hypergraph $\cH$. This completes the proof.
\end{proof}
\section{Identifying Noisy Constraints in Spread Bipartite \texorpdfstring{$k$}{k}-XOR}
\label{sec:bipartite-kXOR}
In this section, we prove \cref{lem:algbipartitexor}. The proof will be decomposed into the following steps.
First, we take the semirandom planted bipartite $k$-XOR instance $\psi$ and transform it into a $2$-XOR instance $\phi$. Second, we decompose the constraint graph of $\phi$ into expanders. For each expander in the decomposition, we argue that the SDP solution to this subinstance is rank $1$, and moreover agrees \emph{exactly} with the planted assignment. This allows us to identify, for each expanding subinstance, \emph{exactly} which edges in $\phi$ are errors. Finally, we use this information to identify the set of corrupted constraints in the original instance $\psi$, which finishes the proof.

\subsection{Setup and key notation}
\label{sec:setup}
We now introduce the key notation that shall be used throughout this section. Let $\psi$ be the semirandom $\tau$-spread $p$-bipartite $k$-XOR instance (recall \cref{def:spread-bipartite-kXOR}) with $m$ constraints given as the input to the algorithm. Recall that the instance $\psi$ is specified by a collection of $p$ hypergraphs $\{\cH_u\}_{u \in [p]}$, where each $\cH_u$ is a $(k-1)$-uniform hypergraph on $n$ vertices and $|\cH_u| = m/p$. Each constraint in $\psi$ is specified by a pair $(u,C)$ where $u \in [p]$, $C \in \cH_u$, and has a right-hand side $b_{u,C} \in \Fits$, and the constraints are $y_u \prod_{i \in C} x_i = b_{u,C}$, where $\{y_u\}_{u \in [p]}$ and $\{x_i\}_{i \in [n]}$ are variables. Because the instance $\psi$ is semirandom with noise parameter $\eta$ and planted assignment $(x^*, y^*)$, for each constraint $(u, C)$ we have, with probability $1 - \eta$ independently, $b_{u,C} = y^*_u \prod_{i \in C} x^*_i$, and otherwise $b_{u,C} = -y^*_u \prod_{i \in C} x^*_i$.
Our goal is to output, in $n^{O(k)}$-time, a set $\Alg_1(\cH)$ of size $\leq \tau m$ to discard, and then for the rest of the instance, identify exactly the corrupted constraints, i.e., those for which $b_{u,C} = -y^*_u \prod_{i \in C} x^*_i$.

We now define the $2$-XOR instance $\phi$ from $\psi$. An example is shown in \cref{fig:constraint-graph-example}.

\begin{definition}[$2$-XOR instance $\phi$ from bipartite $k$-XOR $\psi$]
\label{def:2XOR-from-bipartite-kXOR}
    For every $u \in [p]$ and $\cH_u$, we partition $\cH_u$ arbitrarily into two sets $\cH_u^{(L)}$ and $\cH_u^{(R)}$ of equal size.
    \begin{itemize}
        \item If $k$ is odd, then there are $\binom{n}{\frac{k-1}{2}}^2$ variables in $\phi$, one variable $z_{(S_1, S_2)}$ for each pair of sets $S_1, S_2 \subseteq [n]$ where $|S_1| = |S_2| = \frac{k-1}{2}$.
        
        \item If $k$ is even, then there are $2 \binom{n}{\ceil{\frac{k-1}{2}}} \binom{n}{\floor{\frac{k-1}{2}}}$ variables in $\phi$, one variable $z_{(S_1, S_2)}$ for each pair of sets $S_1, S_2 \subseteq [n]$ where either $\abs{S_1} = \ceil{\frac{k-1}{2}}$ and $\abs{S_2} = \floor{\frac{k-1}{2}}$ or $\abs{S_1} = \floor{\frac{k-1}{2}}$ and $\abs{S_2} = \ceil{\frac{k-1}{2}}$.
    \end{itemize}
    For each $u\in [p]$,  $C \in \cH_u^{(L)}$ and $C' \in \cH_u^{(R)}$, we arbitrarily partition $C$ into sets $S_1 \cup S_2$ and $C'$ into sets $S'_1 \cup S'_2$, where $\abs{S_1} = \abs{S'_1} = \ceil{\frac{k-1}{2}}$ and $\abs{S_2} = \abs{S'_2} = \floor{\frac{k-1}{2}}$. We then add the constraint $z_{(S_1,S'_2)} z_{(S_2,S_1')} = b_{u,C} b_{u,C'}$ to $\phi$.
\end{definition}
It is intuitive to think of clauses from $\cH_u^{(L)}$ and $\cH_u^{(R)}$ as having different colors, and each variable $z_{(S_1,S_2')}$ contains roughly $k/2$ of each color.
See \cref{fig:constraint-graph-example} for an example of a 2-XOR $\phi$ constructed from a bipartite $k$-XOR $\psi$.

\begin{observation}[Size of $\phi$]
\label{obs:edges-in-phi}
    The number of variables in $\phi$ is at most $n^{k-1}$ (for both even and odd $k$).
    Since each $|\cH_u| = m/p$, $\abs{\cH_u^{(L)}} = \abs{\cH_u^{(R)}} = \frac{m}{2p}$, and the number of constraints in $\phi$ is exactly $p \cdot (\frac{m}{2p})^2 = \frac{m^2}{4p}$.
    In particular, when $m \geq n^{\frac{k-1}{2}} \sqrt{p} \cdot \beta$ for $\beta = \poly(\log n)$ as assumed in \cref{lem:algbipartitexor}, the average degree of $\phi$ is at least $\frac{1}{4}\beta^2$.
\end{observation}

\begin{remark}[Corrupted constraints in $\phi$]
    A constraint $z_{(S_1,S'_2)} z_{(S_2,S_1')} = b_{u,C} b_{u,C'}$ in $\phi$ is \emph{corrupted} if exactly one of $b_{u,C}$ and $b_{u,C'}$ is corrupted in $\psi$. Thus, if each constraint in $\psi$ is corrupted with probability $\eta \in (0,1/2)$, then each constraint in $\phi$ is corrupted with probability $2\eta(1-\eta) < 1/2$. Note, however, that the constraints in $\phi$ are not corrupted independently.
\end{remark}

We need some more definitions about the constraint graph of $\phi$.

\begin{definition}[Constraint graph of $\phi$]
    Let $G(\phi) = (V, E)$ be the constraint graph of $\phi$. Notice that each edge $e\in E$ uniquely identifies $u(e) \in [p]$ and $C_L(e)\in \cH_{u(e)}^{(L)}$, $C_R(e) \in \cH_{u(e)}^{(R)}$.
    For each $u \in [p]$, $C \in \cH_u^{(L)}$, define $G_{u,C}^{(L)}(\phi)$ to be the subgraph of $G$ that $C$ participates in, i.e., with edge set $\Set{e \in E: u(e) = u,\ C_L(e) = C}$.
    We similarly define $G_{u,C'}^{(R)}(\phi)$ for $C' \in \cH_u^{(R)}$.
\end{definition}

\begin{figure}[ht!]
    \centering
    \includegraphics[width=0.9\textwidth]{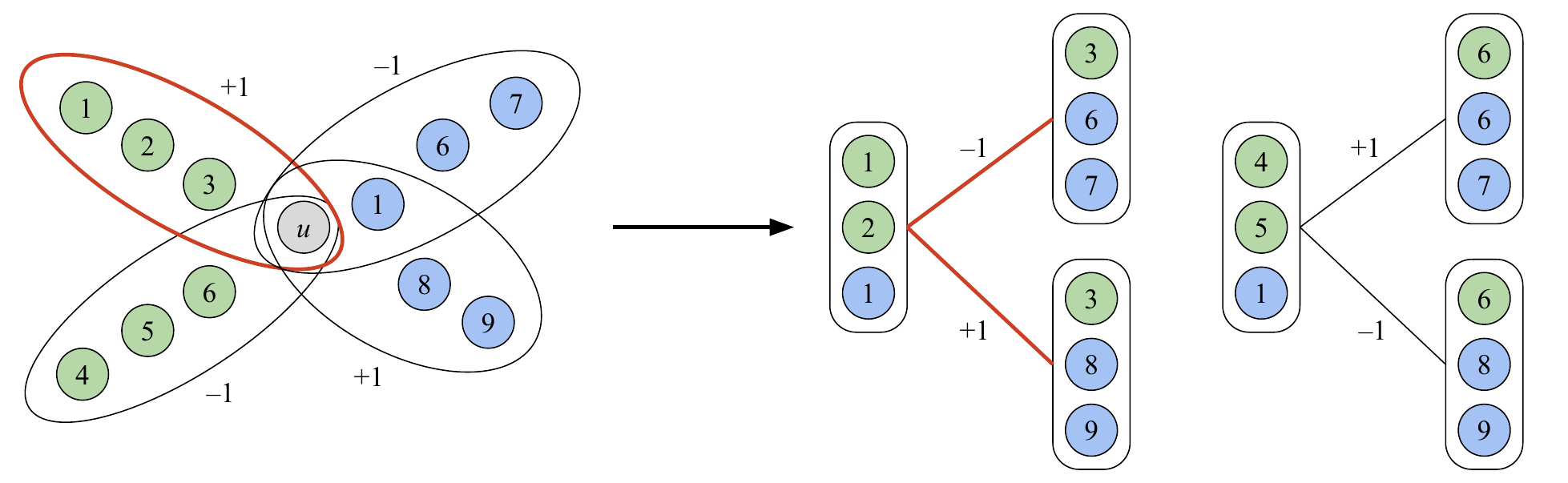}
    \caption{An example of the 2-XOR instance $\phi$ from a bipartite $4$-XOR $\psi$ (\cref{def:2XOR-from-bipartite-kXOR}). On the left, $\cH_u^{(L)}$ consists of $C_1 = \{1,2,3\}$ and $C_2 = \{4,5,6\}$ (with green vertices), and $\cH_u^{(R)}$ consists of $C_1' = \{1,6,7\}$ and $C_2' = \{1,8,9\}$ (with blue vertices).
    On the right, the constraint graph $G(\phi)$ has vertices $z_{S_1,S_2}$ where either $|S_1| = 2$, $|S_2|=1$ or $|S_1|=1$, $|S_2|=2$ (we can view $S_1$, $S_2$ as having green, blue vertices).
    Each edge corresponds to two clauses in $\psi$; for example, the edge $\Set{z_{\{1,2\},\{1\}}, z_{\{3\},\{6,7\}} }$ comes from the clauses $C_1$ and $C_1'$. \\
    \textbf{Corruptions.} In the figure, we label a clause $-1$ if it is corrupted and $+1$ otherwise. An edge in $G$ is corrupted if exactly one of the two corresponding clauses in $\psi$ is corrupted. \\
    \textbf{Degree of $G_{u,C}^{(L)}(\phi)$.} For $C_1 \in \cH_u^{(L)}$, the subgraph $G_{u,C_1}^{(L)}(\phi)$ corresponds to the edges colored red, i.e., all edges that $C_1$ participates in. The vertex $z_{\{1,2\},\{1\}}$ has degree 2 in $G_{u,C_1}^{(L)}(\phi)$ because $|C_1' \cap C_2'| = 1$.
    }
    \label{fig:constraint-graph-example}
\end{figure}

We next make the important observation that the degree of a vertex in $G_{u,C}^{(L)}(\phi)$ is upper bounded by the number of $C' \in \cH_u^{(R)}$ sharing at least $\floor{\frac{k-1}{2}}$ vertices.
See \cref{fig:constraint-graph-example} also for an illustration.
Therefore, assuming that $\psi$ is $\tau$-spread, we have a maximum degree bound on $G_{u,C}^{(L)}(\phi)$ and $G_{u,C'}^{(R)}(\phi)$ for all $u\in[p]$, $C \in \cH_u^{(L)}$ and $C'\in \cH_u^{(R)}$.

\begin{lemma}[Degree bounds for $G_{u,C}^{(L)}$, $G_{u,C'}^{(R)}$]
\label{lem:degree-bounds-eps-spread}
    Let $\psi$ be an $\tau$-spread $p$-bipartite $k$-XOR instance. Then, for any $u\in [p]$, $C \in \cH_u^{(L)}$ and $C'\in \cH_u^{(R)}$, the maximum degree of $G_{u,C}^{(L)}(\phi)$, $G_{u,C'}^{(R)}(\phi)$ is at most $1/\tau^2$.
\end{lemma}
\begin{proof}
    Consider any $C\in \cH_u^{(L)}$ and two adjacent edges $\set{z_{(S_1,S_2')}, z_{(S_2,S_1')}}$ and $\set{z_{(S_1,S_2'')}, z_{(S_2,S_1'')}}$ in $G_{u,C}^{(L)}(\phi)$ formed by joining $C = S_1 \cup S_2$ with $C' = S_1' \cup S_2'$ and $C'' = S_1'' \cup S_2'' \in \cH_u^{(R)}$. As the edges are adjacent, it must be the case that either $S_1' = S_1''$ or $S_2' = S_2''$, which means that $\abs{C' \cap C''} \geq \floor{\frac{k-1}{2}}$. Thus, the degree of a vertex $z_{(S_1,S_2')}$ in $G$ is upper bounded by the maximum number of $C' \in \cH_u^{(R)}$ that all share the same $\floor{\frac{k-1}{2}}$ variables.

    Suppose $\psi$ is $\tau$-spread, meaning that $\deg_u(Q) \leq \frac{1}{\tau^2} \max(1, n^{\frac{k}{2}-1-|Q|})$ for $Q \subseteq [n]$. Since $\frac{k}{2}-1- \floor{\frac{k-1}{2}} \leq 0$, we have that $G_{u,c}^{(L)}(\phi)$ has maximum degree $\leq 1/\tau^2$.
\end{proof}


\subsection{Proof outline}
\label{sec:bipartite-kXOR-outline}

With the setup in \cref{sec:setup} in hand, our proof now proceeds in three conceptual steps.

\parhead{Step 1: graph pruning and expander decomposition.}
Suppose the instance $\phi$ has average degree $d$. We first prune the instance using \cref{lem:graph-pruning} such that the resulting constraint graph has minimum degree $\geq \eps d$ while only removing $\eps$ fraction of the constraints, where $\eps = o(1)$.
We further apply expander decomposition (\cref{fact:expander-decomp}) to the pruned instance to obtain subinstances $\phi_1, \dots, \phi_T$ while discarding only a $\eps$ fraction of the constraints of $\phi$ such that the constraint graph of each $\phi_i$ has spectral gap $\wt{\Omega}(\eps^2)$.

\parhead{Step 2: relative spectral approximation and recovery of corrupted pairs.} We show that for each expanding subinstance $\phi_i$, the basic SDP for the $2$-XOR instance $\phi_i$ is equal to $x^* (x^*)^{\top}$, where $x^*$ is the planted assignment for $\phi$. That is, the SDP solution is \emph{rank $1$} and agrees with the \emph{planted assignment} for $\phi$. We show this by arguing that, for each $\phi_i$, the Laplacian of the corrupted constraints in $\phi_i$ is a \emph{spectral sparsifier} of the Laplacian of the constraint graph of $\phi_i$ (see \cref{lem:sdpuniqueness}). Here, we crucially use that each such constraint graph has large minimum degree and spectral gap.

From this, it is trivial to identify the corrupted edges in each $\phi_i$, as they are the ones violated by the SDP solution. We are not quite done yet, however, because each constraint in $\phi$ corresponds to a \emph{pair} of constraints in the original instance $\psi$.

\parhead{Step 3: recovery of corrupted constraints from corrupted pairs.} The previous step shows that for all but a $\eps$ fraction of tuples $(u,C,C')$ where $u \in [p]$, $C \in \cH_u^{(L)}$, and $C' \in \cH_u^{(R)}$, we can recover the product $\xi_u(C)\xi_u(C')$, where $\xi_u(C) = -1$ if $(u,C)$ is noisy in $\psi$, and is $+1$ otherwise. Because $\eps$ is small, it must be the case that for most $u \in [p]$, we know the product $\xi_u(C)\xi_u(C')$ (from Step 2) for \emph{most} pairs $(C,C')$ with $C \in \cH_u^{(L)}$ and $C' \in \cH_u^{(R)}$.

Suppose we knew $\xi_u(C)\xi_u(C')$ for all $(C,C') \in \cH_u^{(L)} \times \cH_u^{(R)}$. Then, it is trivial to decode $\xi_u(C)$ \emph{up to a global sign}. Formally, we could obtain $z \in \Fits^{\cH_u}$ where $z_C = \alpha \xi_u(C)$ for some $\alpha \in \Fits$. From this, it is easy to obtain $\xi_u(C)$, as the fraction of $C \in \cH_u$ for which $\xi_u(C) = -1$ should be roughly $\eta < \frac{1}{2}$; so, if $z$ has $< \frac{1}{2}$-fraction of $-1$'s, then $z = \xi_u(C)$, and otherwise $-z = \xi_u(C)$. This, however, requires $|\cH_u| \geq \Omega\Bigparen{\frac{\log n}{(1-2\eta)^2}}$ for a high-probability result.

Additionally, we do not quite know $\xi_u(C)\xi_u(C')$ for all $(C,C') \in \cH_u^{(L)} \times \cH_u^{(R)}$: we only know this for all but a $\eps_u$-fraction of the pairs. By forming a graph $G_u$ where we have an edge $(C,C')$ if $(C,C')$ is a pair where we know $\xi_u(C)\xi_u(C')$, we can thus obtain such a $z$ for all $C$ in the largest connected component of $G_u$. Because $G_u$ is obtained by taking a \emph{complete biclique} and deleting only a $\eps_u$-fraction of all edges, the largest connected component has size $(1-\eps_u)\abs{\cH_u}$, and so we can recover $\xi_u(C)$ for all but a $\eps_u$-fraction of constraints in $\cH_u$. We do this for each partition $u$, which finishes the proof.
\subsection{Graph pruning and expander decomposition}

This step is a simple combination of graph pruning and expander decomposition.

\begin{lemma} \label{lem:pruning-expander-decomp}
    Fix $\eps \in (0,1)$.
    There is a polynomial-time algorithm such that, given a 2-XOR instance $\phi$ whose constraint graph has $m$ edges and average degree $d$, outputs subinstances $\phi_1,\dots,\phi_T$ on disjoint variables with the following guarantees:
    $\phi_1,\dots,\phi_T$ contain at least $1- \eps$ fraction of the constraints in $\phi$, and for each $i\in [T]$, the constraint graph $G_i$ of $\phi_i$, after adding some self-loops, has minimum degree at least $\frac{1}{3}\eps d$ and $\lambda_2(\wt{L}_{G_i}) \geq \Omega(\eps^2/\log^2 m)$.
\end{lemma}

The self-loops in \cref{lem:pruning-expander-decomp} are only for the analysis of $\wt{L}_{G_i}$ and do not correspond to actual constraints in $\phi_i$.
Observe that adding self-loops to a graph $G$ does not change the \emph{unnormalized} Laplacian $L_G$, but as $D_G$ (the degree matrix) increases, the spectral gap of the \emph{normalized} Laplacian, i.e.\ $\lambda_2(\wt{L}_G) = \lambda_2(D_G^{-1/2} L_G D_G^{-1/2})$, may decrease.
The expander decomposition algorithm (\cref{fact:expander-decomp}) guarantees that each piece, even after adding self-loops to preserve degrees, has large spectral gap.
This does not change the subinstances $\phi_1,\dots,\phi_T$, but in the next section, it is crucial that we use this stronger guarantee to ensure a lower bound on the minimum degree.

\begin{proof}[Proof of \cref{lem:pruning-expander-decomp}]
    We first apply the graph pruning algorithm (\cref{lem:graph-pruning}) such that the resulting instance has minimum degree $\geq \frac{\eps}{3} d$ and at least $(1-\frac{2}{3}\eps)m$ constraints.
    Then, we apply expander decomposition (\cref{fact:expander-decomp}) that partitions the vertices of the pruned graph $G'$ into $V_1,\dots, V_T$ such that the number of edges across partitions is at most $\frac{\eps}{3} m$, and for each $i\in[T]$, the normalized Laplacian satisfies $\lambda_2( \wt{L}_{G'\{V_i\}}) \geq \Omega(\eps^2/\log^2 m)$.
    Here we recall that $G'\{V_i\}$ is the induced subgraph of $G'$ with self-loops such that the vertices in $G'\{V_i\}$ have the same degrees as in $G'$.
    
    In total, we have removed at most $\eps m$ edges. This completes the proof.
\end{proof}


\subsection{Rank-1 SDP solution from expansion and relative spectral approximation}

We next show that for each subinstance $\phi_i$ obtained from \cref{lem:pruning-expander-decomp}, its constraint graph $G$ and the subgraph of corrupted edges $H$ satisfy $L_H \prec \frac{1}{2} L_G$.
Recall from \cref{lem:sdpuniqueness,lem:spectral-sparsification} that this implies the basic SDP for the 2-XOR $\phi_i$ is rank $1$ and agrees with the planted assignment of $\phi$.

The next lemma is analogous to \cref{lem:spectral-sparsification} but differs in an important way: a constraint in $\phi$ is corrupted if and only if exactly one of the two corresponding constraints in $\psi$ is corrupted; thus, the corruptions in $\phi$ are \emph{correlated}.
This is why each constraint in $\phi$ is obtained from one clause in $\cH_u^{(L)}$ and one clause in $\cH_u^{(R)}$ (recall \cref{def:2XOR-from-bipartite-kXOR}), so that in the proof below we have independent randomness to perform a ``2-step sparsification'' proof.
It is also worth noting that the following lemma requires not just a lower bound on the minimum degree and spectral gap of $G$ but also that the original bipartite $k$-XOR instance $\psi$ is \emph{well-spread}, which allows us to apply \cref{lem:degree-bounds-eps-spread}.

Same as \cref{lem:spectral-sparsification}, the following lemma is a purely graph-theoretic statement.

\newcommand{\xiL}[1]{\xi^{(1)}_{#1}}
\newcommand{\xiR}[1]{\xi^{(2)}_{#1}}
\newcommand{\GL}[1]{G^{(1)}_{#1}}
\newcommand{\GR}[1]{G^{(2)}_{#1}}

\begin{lemma}[Relative spectral approximation with correlated subsamples]
\label{lem:spectral-sparsification-odd}
    Suppose $G = (V,E)$ is an $n$-vertex graph with minimum degree $d_{\min}$ (self-loops and parallel edges allowed) and spectral gap $\lambda_2(\widetilde{L}_G) = \lambda > 0$.
    Let $m_1, m_2 \in \N$, $\eta \in [0,1/2)$, and let $\xiL{1}, \dots, \xiL{m_1}, \xiR{1},\dots,\xiR{m_2}$ be i.i.d.\ random variables that take value $-1$ with probability $\eta$ and $+1$ otherwise.
    Suppose there is an injective map that maps each edge $e \mapsto (c_1(e), c_2(e)) \in [m_1] \times [m_2]$, and for each $i\in [m_1]$ (resp.\ $j\in [m_2]$) define $\GL{i}$ (resp.\ $\GR{j}$) be the subgraph of $G$ with edge set $\{e\in E: c_1(e) = i\}$ (resp.\ $\{e\in E: c_2(e)=j\}$).
    Moreover, suppose $\GL{i}$ and $\GR{j}$ have maximum degree $\leq \Delta$ for all $i\in [m_1]$, $j\in [m_2]$.

    Let $H$ be the subgraph of $G$ with edge set $\bigset{e\in E: \xiL{c_1(e)} \xiR{c_2(e)} = -1}$.
    There is a universal constant $B > 0$ such that if $d_{\min}\lambda \geq B \Delta \log n$, then with probability $1-O(n^{-2})$,
    \begin{align*}
        L_H \preceq \max\Paren{ (1+\delta) \cdot 2\eta(1-\eta),\ \frac{1}{3}} \cdot L_G
    \end{align*}
    for $\delta = \sqrt{\frac{B \Delta \log n}{d_{\min}\lambda}}$.
\end{lemma}

Let $\gamma \defeq 1- 2\eta > 0$ since $\eta < \frac{1}{2}$.
Notice that $2\eta(1-\eta) = \frac{1}{2}(1 - \gamma^2)$, which approaches $\frac{1}{2}$ as $\eta \to \frac{1}{2}$.
Thus, if $\delta \leq \gamma^2$, then $(1+\delta) \cdot 2\eta(1-\eta) \leq (1+\gamma^2) \cdot \frac{1}{2}(1-\gamma^2) < \frac{1}{2}$, and $L_H \prec \frac{1}{2}L_G$ suffices to conclude via \cref{lem:sdpuniqueness} that the SDP relaxation on the expanding subinstance is rank 1 and recovers the planted assignment, which also gives us the set of corrupted constraints.

\begin{proof}[Proof of \cref{lem:spectral-sparsification-odd}]
    First, note that by the definition of Laplacian and the spectral gap of $L_G$, $\spn(\vec{1})$ is exactly the null space of $L_G$ and is contained in the null space of $L_H$.
    Therefore, recalling that $L_G = D_G^{1/2} \wt{L}_G D_G^{1/2}$, it suffices to prove that
    \begin{equation} \label{eq:L_H-norm-condition}
        \Norm{ (\wt{L}_G^{\dagger})^{1/2} D_G^{-1/2} L_H D_G^{-1/2} (\wt{L}_G^{\dagger})^{1/2} }_2 \leq \max\Paren{ (1+\delta) \cdot 2\eta(1-\eta),\ \frac{1}{3}}  \mper
    \end{equation}
    Here $\wt{L}_G^{\dagger}$ is the pseudo-inverse of $\wt{L}_G$, and $\|\wt{L}_G^{\dagger}\|_2 \leq 1/\lambda$.
    For simplicity, for any graph $G'$, we will write $\widehat{L}_{G'} \coloneqq (\wt{L}_G^{\dagger})^{1/2} D_G^{-1/2} L_{G'} D_G^{-1/2} (\wt{L}_G^{\dagger})^{1/2}$. Thus,
    \begin{equation*}
        \wh{L}_H = \sum_{e\in E} \1\Paren{\xiL{c_1(e)} \xiR{c_2(e)} = -1} \cdot \wh{L}_e \mcom
        \text{ and } \E[\wh{L}_H] = 2\eta (1-\eta) \sum_{e\in E} \wh{L}_e  \mper
    \end{equation*}
    Note that $\sum_{e\in E} \wh{L}_e = \wh{L}_G$, a projection matrix, thus $\bignorm{\sum_{e\in E}\wh{L}_e}_2 = 1$.
    
    For each $i\in [m_1]$, we further define $\GL{i,+}$ and $\GL{i,-}$ to be (random) edge-disjoint subgraphs of $\GL{i}$ where $\GL{i,+}$ has edge set $\bigset{e\in E: c_1(e)=i, \xiR{c_2(e)} = +1}$ and $\GL{i,-}$ has edge set $\bigset{e\in E: c_1(e)=i, \xiR{c_2(e)} = -1}$.
    Note that $\GL{i,+}$, $\GL{i,-}$ are independent of $\xiL{} = (\xiL{1},\dots,\xiL{m_1})$.
    By the maximum degree bound on $\GL{i}$, we have that $\bignorm{L_{\GL{i,+}}}_2$ and $\bignorm{L_{\GL{i,-}}}_2 \leq \bignorm{L_{\GL{i}}}_2 \leq 2\Delta$. Thus,
    \begin{equation}  \label{eq:hat-L_Gi-norm}
        \Bignorm{\wh{L}_{\GL{i,+}}}_2,\ \Bignorm{\wh{L}_{\GL{i,-}}}_2 
        \leq \Bignorm{\wh{L}_{\GL{i}}}_2 
        \leq 2\Delta \cdot \Norm{\wt{L}_G^\dagger}_2 \cdot \Norm{D_G^{-1}}_2 \leq \frac{2\Delta}{d_{\min}\lambda} \mper
    \end{equation}
    Similarly, for $j\in [m_2]$, $\GR{j,+}$ and $\GR{j,-}$ are (random) edge-disjoint subgraphs of $\GR{j}$ independent of $\xiR{} = (\xiR{1},\dots,\xiR{m_2})$ such that $\bignorm{\wh{L}_{\GR{j,+}}}_2$ and $\bignorm{\wh{L}_{\GR{j,-}}}_2 \leq \frac{2\Delta}{d_{\min}\lambda}$.

    Now, we first fix $\xiR{} \in \Fits^{m_2}$. Observe that we can write $\wh{L}_H$ as
    \begin{equation} \label{eq:hat-L_H}
        \wh{L}_H = \sum_{i\in [m_1]} \1(\xiL{i} = +1) \cdot \wh{L}_{\GL{i,-}} + \1(\xiL{i} = -1) \cdot \wh{L}_{\GL{i,+}} \mcom
    \end{equation}
    and
    \begin{equation}  \label{eq:hat-L_H-cond-z}
    \begin{aligned}
        \E[\wh{L}_H | \xiR{}] &= (1-\eta) \sum_{i\in [m_1]} \wh{L}_{\GL{i,-}} + \eta \sum_{i\in[m_1]} \wh{L}_{\GL{i,+}} \\
        &= \sum_{e\in E} \Bigparen{(1-\eta) \cdot \1(\xiR{c_2(e)}=-1) + \eta \cdot \1(\xiR{c_2(e)} = +1)} \cdot \wh{L}_e  \\
        &\coloneqq \sum_{e\in E} w_{c_2(e)} \cdot \wh{L}_e \mper
    \end{aligned}
    \end{equation}
    Here $w_{c_2(e)} \in \{\eta, 1-\eta\}$, thus $\bignorm{\E[\wh{L}_H | \xiR{}]}_2 \geq \eta \bignorm{\sum_{e\in E} \wh{L}_e}_2 = \eta$.

    We now split the analysis into two cases. Let $\eta_0 \defeq 1/12$.

    \parhead{Case 1: $\eta \geq \eta_0$.}
    
    In light of \cref{eq:hat-L_H}, we define $X_i \coloneqq \1(\xiL{i} = +1) \cdot \wh{L}_{\GL{i,-}} + \1(\xiL{i} = -1) \cdot \wh{L}_{\GL{i,+}}$ such that $\wh{L}_H = \sum_{i\in [m_1]} X_i$.
    Moreover, we have that $X_i \succeq 0$ and $\|X\|_2 \leq \frac{2\Delta}{d_{\min}\lambda}$ almost surely from \Cref{eq:hat-L_Gi-norm}.
    Thus, applying matrix Chernoff (\Cref{fact:matrix-chernoff}), we get
    \begin{equation} \label{eq:first-matrix-chernoff}
    \begin{aligned}
        \Pr_{\xiL{}} \Brac{\Bignorm{\wh{L}_H}_2 \geq (1+\delta) \Norm{\E[\wh{L}_H | \xiR{}]}_2 } &\leq n \cdot \exp\Paren{-\frac{1}{3}\delta^2 \Norm{\E[\wh{L}_H | \xiR{}]}_2 \cdot \frac{d_{\min}\lambda}{2\Delta}} \\
        &\leq n \cdot \exp\Paren{-\frac{\delta^2 \eta d_{\min} \lambda}{6 \Delta}} \mcom
    \end{aligned}
    \end{equation}
    which is at most $O(n^{-2})$ as long as $\delta^2 \geq  \frac{B_1\Delta \log n}{d_{\min}\lambda}$ for a large enough constant $B_1$.

    Next, we similarly prove concentration for $\bignorm{\E[\wh{L}_H | \xiR{}]}_2$ over $\xiR{}$. Recalling \Cref{eq:hat-L_H-cond-z},
    \begin{equation*}
        \E[\wh{L}_H|\xiR{}] = \sum_{e\in E} w_{c_2(e)} \cdot \wh{L}_e
        = \sum_{j\in [m_2]} w_j \sum_{e\in \GR{j}} \wh{L}_e = \sum_{j\in [m_2]} w_j \cdot \wh{L}_{\GR{j}} \mper
    \end{equation*}
    $\E[w_j] = 2\eta(1-\eta)$, and $\bignorm{\E_{\xiR{}} \E[\wh{L}_H|\xiR{}]}_2 = 2\eta(1-\eta) \bignorm{\sum_{e\in E} \wh{L}_e}_2 = 2\eta(1-\eta)$.
    Since $\bignorm{w_j \wh{L}_{\GR{j}}}_2 \leq \frac{2(1-\eta)\Delta}{d_{\min}\lambda}$, we can apply matrix Chernoff again:
    \begin{equation} \label{eq:second-matrix-chernoff}
        \Pr_{\xiR{}} \Brac{\Norm{\E[\wh{L}_H|\xiR{}]}_2 \geq (1+\delta') \cdot 2\eta(1-\eta)} \leq n \cdot \exp\Paren{-\frac{1}{3}\delta'^2 \cdot 2\eta(1-\eta) \cdot \frac{d_{\min}\lambda}{2(1-\eta)\Delta}}
    \end{equation}
    which is at most $O(n^{-2})$ as long as $\delta'^2 \geq \frac{B_2\Delta \log n}{d_{\min}\lambda}$ for a large enough constant $B_2$.
    Combining both tail bounds, by the union bound, we have that with probability at least $1 - O(n^{-2})$, $\bignorm{\wh{L}_H}_2 \leq (1+\delta) \cdot 2\eta(1-\eta)$ as long as $\delta^2 \geq \frac{B\Delta \log n}{d_{\min}\lambda}$ for a large enough $B$.
    This establishes \Cref{eq:L_H-norm-condition}, proving the lemma for this case.

    \parhead{Case 2: $\eta < \eta_0$.} To handle this case, observe that the exact same analysis goes through for $\wt{H} = \{e\in E: \xiL{c_1(e)} = -1 \text{ or } \xiR{c_2(e)} = -1\} \supseteq H$.
    Indeed, similar to \cref{eq:hat-L_H,eq:hat-L_H-cond-z}, we have $\wh{L}_{\wt{H}} = \sum_{i\in[m_1]} \wt{X}_i$ where $\wt{X}_i = \1(\xiL{i} = +1) \cdot \wh{L}_{\GL{i,-}} + \1(\xiL{i} = -1) \cdot \wh{L}_{\GL{i}}$ (notice the 2nd term is $\GL{i}$ instead of $\GL{i,+}$), and 
    \begin{align*}
        \E[\wh{L}_{\wt{H}} | \xiR{}] = (1-\eta) \sum_{i\in [m_1]} \wh{L}_{\GL{i,-}} + \eta \sum_{i\in[m_1]} \wh{L}_{\GL{i}}
        = \sum_{e\in E} \wt{w}_{c_2(e)} \cdot \wh{L}_e 
        = \sum_{j\in[m_2]} \wt{w}_j \cdot \wh{L}_{\GR{j}} \mcom
    \end{align*}
    where $\wt{w}_j = 1$ if $\xiR{j} = -1$ and $\eta$ if $\xiR{j} = +1$, hence $\E[\wt{w}_j] = \eta + \eta(1-\eta) = \eta(2-\eta)$.
    Moreover, $\bignorm{\E_{\xiR{}} \E[\wh{L}_{\wt{H}}|\xiR{}]}_2 = \eta(2-\eta) \bignorm{\sum_{e\in E} \wh{L}_e}_2 = \eta(2-\eta)$.

    First, set $\eta = \eta_0$, and let $\wt{H}_0$ be the random subgraph as defined above. Similar to \cref{eq:first-matrix-chernoff,eq:second-matrix-chernoff}, we apply matrix Chernoff (\cref{fact:matrix-chernoff}) and get that with probability $1-O(n^{-2})$, $\bignorm{\wh{L}_{\wt{H}_0}}_2 \leq (1+\delta) \cdot \eta_0(2-\eta_0)$ for $\delta = \sqrt{\frac{B\Delta\log n}{d_{\min}\lambda}} \leq 1$.
    In particular, this means that $L_{\wt{H}_0} \preceq 2\eta_0(2-\eta_0) L_G \preceq \frac{1}{3} L_G$ when $\eta_0 = 1/12$.
    
    Now, fix any $\eta < \eta_0$.
    We can obtain a coupling between this case and the case when $\eta = \eta_0$ by randomly changing $\xiL{i}$ and $\xiR{j}$ from $+1$ to $-1$ (while not flipping the ones with $-1$).
    Notice that $\wt{H}$ is monotone increasing as we change any $+1$ to $-1$ (whereas $H$ is not!), thus we must have $\wt{H} \subseteq \wt{H}_0$ in this coupling. Then, as $H \subseteq \wt{H}$, we have
    \begin{flalign*}
        L_H \preceq L_{\wt{H}} \preceq L_{\wt{H}_0} \preceq \frac{1}{3} L_G
    \end{flalign*}
    with probability $1-O(n^{-2})$. This finishes the proof of \cref{lem:spectral-sparsification-odd}.
\end{proof}

\subsection{Recovery of corrupted constraints from corrupted pairs}

We have thus shown that, with probability $\geq 1 - 1/\poly(n)$, we can \emph{exactly} recover the set of corrupted constraints within each expanding subinstance $\phi_1, \dots, \phi_T$. Recall that after pruning and expander decomposition (\cref{lem:pruning-expander-decomp}), the expanding subinstances contain a $(1 - \eps)$-fraction of all edges in the instance $\phi$, and the set of edges removed only depends on the constraint graph and not the right-hand sides of $\phi$. As stated in \cref{obs:edges-in-phi}, the instance $\phi$ has exactly $m^2/4p$ edges, and they correspond exactly to the set $\{(u,C,C') : u \in [p], C \in \cH_{u}^{(L)}, C' \in \cH_{u}^{(R)}\}$, and moreover an edge $e$ in $\phi$ is corrupted if and only if exactly one of the two constraints $(u,C),(u,C')$ is corrupted in the original instance $\psi$, where $e$ corresponds to $(u,C,C')$. For each $u \in [p]$ and $C \in \cH_u = \cH_u^{(L)} \cup \cH_u^{(R)}$, let $\xi_u(C) = -1$ if $(u,C)$ is corrupted in $\psi$, and $1$ otherwise. It thus follows that we have learned, for $1 - \eps$ fraction of all $\{(u,C,C') : u \in [p], C \in \cH_{u}^{(L)}, C' \in \cH_{u}^{(R)}\}$, the product $\xi_u(C) \cdot \xi_u(C')$.

It now remains to show how to recover $\xi_u(C)$ for most $u \in [p]$, $C \in \cH_u$. For each $u \in [p]$, let $P_u \subseteq \{(C,C') : C \in \cH_{u}^{(L)}, C' \in \cH_{u}^{(R)}\}$ such that we have determined $\xi_u(C) \cdot \xi_u(C')$, and let $P = \cup_{u \in [p]} P_u$. We know that $\abs{P} \geq (1 - \eps) \frac{m^2}{4p}$. Let $\eps_u$ be chosen so that $\abs{P_u} = (1 - \eps_u) \frac{m^2}{4p^2}$, i.e., $\eps_u$ is the fraction of pairs in $\cH_u^{(L)} \times \cH_u^{(R)}$ that were deleted in \cref{lem:pruning-expander-decomp}. Notice that we have
\begin{equation}\label{eq:averaging}
\begin{aligned}
&(1 - \eps) \frac{m^2}{4p} \leq \abs{P} = \sum_{u\in [p]} \abs{P_u} = \frac{m^2}{4p^2}\sum_{u\in [p]} (1 - \eps_u) \\
&\implies \frac{1}{p} \sum_{u\in [p]} \eps_u \leq \eps \mper
\end{aligned}
\end{equation}
One can think of this problem as a collection of disjoint \emph{satisfiable} (noiseless) $2$-XOR instances on $P_u$, where each $P_u$ is a biclique ($\frac{m}{2p}$ vertices on each side) with $\eps_u$ fraction of edges are removed.

\begin{mdframed}
  \begin{algorithm}[Recover corrupted constraints from corrupted pairs]
    \label{alg:recovery}\mbox{}
    \begin{description}
    \item[Given:]
        For each $u \in [p]$, a set $P_u \subseteq \cH_u^{(L)} \times \cH_u^{(R)}$ such that $|P_u| = (1-\eps_u) \frac{m^2}{4p^2}$ for $\eps_u\in [0,1]$, along with ``right-hand sides'' $\xi_u(C) \cdot \xi_u(C')$ for each $(C,C') \in P_u$.
    \item[Output:]
        For each $u \in [p]$, disjoint subsets $\cA_u^{(1)}, \cA_u^{(2)} \subseteq \cH_u$.
   \item[Operation:]\mbox{}

    \begin{enumerate}
        \item \textbf{Initialize:} $\cA_u^{(1)}, \cA_u^{(2)} = \emptyset$ for each $u \in [p]$.
        \item \textbf{For each $u \in [p]$}:
        \begin{enumerate}
            \item If $\eps_u \geq 1/3$, set $\cA_u^{(1)} = \cH_u$ and $\cA_u^{(2)} = \emptyset$.
            \item Else if $\eps_u < 1/3$, let $G_u$ be the graph with vertex set $\cH_u = \cH_u^{(L)} \cup \cH_u^{(R)}$ with edges given by $P_u$, and let $S_u$ be the size of the largest connected component in $G_u$.
            \item As $S_u$ is connected in $G_u$, and we know $\xi_u(C) \xi_u(C')$ for each edge $(C,C')$ in $G_u$, by solving a linear system of equations we obtain $z \in \Fits^{\cH_u}$ such that either $z_C = \xi_u(C)$ for all $C \in S_u$, or $z_C = -\xi_u(C)$ for all $C \in S_u$. That is, $z_C = \xi_u(C)$ up to a global sign.
            \label{step:find-z}
            \item Pick the global sign to minimize the number of $C \in S_u$ for which $z_C = -1$. Set $\cA_u^{(1)} = \cH_u \setminus S_u$ and $\cA_u^{(2)} = \{C \in S_u : z_C = -1\}$.
            \label{step:pick-sign}
        \end{enumerate}
            
        \item Output $\{\cA^{(1)}_u\}_{u \in [p]}$, $\{\cA^{(2)}_u\}_{u \in [p]}$.
      \end{enumerate}
    \end{description}
  \end{algorithm}
  \end{mdframed}
We now analyze \cref{alg:recovery} via the following lemma.
\begin{lemma}
\label{lem:recovery}
Let $\eta \in [0,1/2)$, and let $|\cH_u| = \frac{m}{p} \geq \frac{24k}{(1-2\eta)^2} \log n$ and $|P_u| = (1-\eps_u) \frac{m^2}{4p^2}$ with $\eps_u \in [0,1]$ for each $u\in[p]$, and $\frac{1}{p} \sum_{u\in[p]} \eps_u \leq \eps$.
The outputs of \cref{alg:recovery} satisfy the following:
\begin{inparaenum}[(1)] \item $\sum_{u\in[p]} |\cA_u^{(1)}| \leq 4\eps m$, and \item with probability $1 - n^{-k}$ over the noise $\{\xi_u(C)\}_{u \in [p], C \in \cH_u}$, for every $u \in [p]$ we have that $\cA_u^{(2)} = \{C \in \cH_u : \xi_u(C) = -1\} \setminus \cA_u^{(1)}$.
\end{inparaenum}
\end{lemma}
\begin{proof}
Suppose that $\eps_u < 1/3$. Observe that $G_u$ is a graph obtained by taking a biclique with left vertices $\cH_u^{(L)}$ and right vertices $\cH_u^{(R)}$, i.e., with $m/2p$ left vertices and $m/2p$ right vertices. The following lemma shows that the largest connected component $S_u$ in $G_u$ has size at least $\frac{m}{p}(1 - \eps_u)$.

\begin{claim}
\label{lem:connectedcomponent}
Let $K_{n,n}$ be the complete bipartite graph with $n$ left vertices $L$ and $n$ right vertices $R$. Let $G$ be a graph obtained by deleting $\eps n^2$ edges from $K_{n,n}$. Then, the largest connected component in $G$ has size $\geq 2n(1-\eps)$.
\end{claim}
 We postpone the proof of \cref{lem:connectedcomponent} to the end of the section, and continue with the proof of \cref{lem:recovery}.

We now argue that we can efficiently obtain the vector $z$ in Step (\ref{step:find-z}) of \cref{alg:recovery}. Indeed, this is done as follows. First, pick one $C_0 \in S_u$ arbitrarily, and set $z_{C_0} = 1$. 
Then, we propagate in a breadth-first search manner: for any edge $(C, C')$ in $S_u$ where $z_C$ is determined, set $z_{C'} = z_C \cdot \xi_u(C) \xi_u(C')$.
We repeat this process until we have labeled all of $S_u$.
Notice that as $S_u$ is a connected component, fixing $z_{C_0}$ for any $C_0 \in S_u$ uniquely determines the assignment of all $S_u$, thus we have obtained $z_C = \xi_u(C)$ up to a global sign.

Now, we observe that $S_u$ does not depend on the noise in $\psi$. Indeed, this is because the pruning and expander decomposition (and thus the graph $G_u$) depends solely on the constraint graph $G$ of the instance $\phi$, and not on the right-hand sides of the constraints. The following lemma thus shows that with high probability over the noise, the number of $C \in S_u$ where $\xi_u(C) = -1$ is strictly less than $1/2 \abs{S_u}$. Hence, in Step (\ref{step:pick-sign}), by picking the assignment $\pm z$ that minimizes the number of $C \in S_u$ with $\xi_u(C) = -1$, we see that $\cA_u^{(2)} = \{C \in S_u : z_C = -1\} = \{C \in S_u : \xi_u(C) = -1\}$.
\begin{claim}
\label{lem:localnoiseconc}
Let $\eta \in (0,1/2)$ be the corruption probability, and assume that $p \leq n^k$ and $\frac{m}{p} \geq \frac{24k}{(1-2\eta)^2} \log n$.
With probability $1 - n^{-k}$ over the noise in $\psi$, it holds that for each $u \in [p]$ with $\eps_u < 1/3$, $\abs{\{C \in S_u : \xi_u(C) = -1\}} <  \frac{1}{2}\abs{S_u}$.
\end{claim}
We postpone the proof of \cref{lem:localnoiseconc}, and finish the proof of \cref{lem:recovery}.
We next bound $\sum_{u\in[p]} |\cA_u^{(1)}|$.
By \cref{eq:averaging} we have that $\frac{1}{p}\sum_{u} \eps_u \leq \eps$. Thus,
\begin{flalign*}
    \sum_{u: \eps_u \geq 1/3} \abs{\cH_u} \leq \frac{m}{p}\sum_{u: \eps_u\geq 1/3} 3\eps_u \leq 3\eps m \mper
\end{flalign*}
Moreover, by \cref{lem:connectedcomponent} we have $|S_u| \geq (1-\eps_u) |\cH_u| = (1-\eps_u) \frac{m}{p}$.
Thus,
\begin{flalign*}
    \sum_{u:\eps_u < 1/3} \abs{\cH_u \setminus S_u} \leq \sum_{u: \eps_u < 1/3} \eps_u \cdot \frac{m}{p} \leq \eps m \mper
\end{flalign*}
Therefore, combining the two,
\begin{flalign*}
    \sum_{u\in [p]} \abs{\cA_u^{(1)}} = \sum_{u : \eps_u < 1/3} \abs{\cH_u \setminus S_u} + \sum_{u : \eps_u \geq 1/3} \abs{\cH_u} \leq 4\eps m \mcom
\end{flalign*}
which finishes the proof of \cref{lem:recovery}.
\end{proof}

In the following, we prove \cref{lem:connectedcomponent,lem:localnoiseconc}.

\begin{proof}[Proof of \cref{lem:connectedcomponent}]
Let $S_1, \dots, S_t$ be the connected components of $G$. Let $\ell_i = \abs{S_i \cap L}$ and $r_i = \abs{S_i \cap R}$. The number of edges in $G$ is at most $\sum_{i = 1}^t \ell_i r_i$.

Now, suppose that the largest connected component of $G$ has size at most $M$. Then, we have that $\ell_i + r_i \leq M$ for all $i \in [t]$. Notice that the number of edges deleted from $K_{n,n}$ to produce $G$ must be at least $n^2 - \sum_{i = 1}^t \ell_i r_i$, and this is at most $\eps n^2$. Hence, by maximizing the quantity $\sum_{i = 1}^t \ell_i r_i$ subject to $\ell_i + r_i \leq M$ for all $i \in [t]$ and $\sum_{i=1}^t \ell_i + r_i = 2n$, we can obtain a lower bound on the number of edges deleted from $K_{n,n}$ in order for the largest connected component of $G$ to have size at most $M$. We have that
\begin{flalign*}
    \sum_{i=1}^t \ell_i r_i \leq \sum_{i=1}^t \Paren{\frac{\ell_i+r_i}{2}}^2 \leq \frac{M}{2}\cdot \sum_{i=1}^t \frac{\ell_i + r_i}{2} = \frac{nM}{2} \mcom
\end{flalign*}
where the first inequality is by the AM-GM inequality. Thus,
\begin{flalign*}
    \eps n^2 \geq n^2 - \frac{nM}{2} \implies M \geq 2n(1 - \eps) \mcom
\end{flalign*}
which finishes the proof.
\end{proof}

\begin{proof}[Proof of \cref{lem:localnoiseconc}]
Let $u$ be such that $\eps_u < 1/3$, and let $S_u$ be the largest connected component in $G_u$. Observe that $S_u$ is determined solely by the constraint graph of $\phi$, and in particular does not depend on the noise in $\phi$ (and hence on the noise in $\psi$).
As $p \leq n^{k}$ by assumption, it thus suffices to show that for each $u \in [p]$, with probability $1 - n^{-2k}$ it holds that $\abs{\{C \in S_u : \xi_u(C) = -1\}} < \frac{1}{2}\abs{S_u}$. Notice that $\abs{\{C \in S_u : \xi_u(C) = -1\}}$ is simply the sum of $\abs{S_u}$ $\text{Bernoulli}(\eta)$ random variables. By Hoeffding's inequality, with probability $\geq 1 - \exp(-2\delta^2 \abs{S_u})$ it holds that $\abs{\{C \in S_u : \xi_u(C) = -1\}} \leq (\eta + \delta) \abs{S_u}$. We choose $\delta = \frac{1}{2}(\frac{1}{2} - \eta)$ such that $\eta + \delta < \frac{1}{2}$ for $\eta \in (0,\frac{1}{2})$. Then, by noting that $2\delta^2 \abs{S_u} \geq 2\delta^2 (1 - \eps_u)\abs{\cH_u} \geq \frac{1}{2}(\frac{1}{2}-\eta)^2 \cdot \frac{2}{3} \cdot \frac{m}{p} \geq 2k \log n$ since $\frac{m}{p} \geq \frac{24k}{(1-2\eta)^2} \log n$, \cref{lem:localnoiseconc} follows.
\end{proof}

\subsection{Finishing the proof of \texorpdfstring{\cref{lem:algbipartitexor}}{Lemma~\ref{lem:algbipartitexor}}}

\begin{proof}[Proof of \cref{lem:algbipartitexor}]
    We are given an $\tau$-spread $p$-bipartite $k$-XOR instance $\psi$ with constraint graph $\cH = \{\cH_u\}_{u\in [p]}$, where we recall from \cref{def:spread-bipartite-kXOR} that
    \begin{inparaenum}[(1)]
        \item $m = |\cH|$ and each $|\cH_u| = \frac{m}{p} \geq 2\floor{\frac{1}{2\tau^2}}$ and $\frac{m}{p}$ is even, and
        \item for any $Q\subseteq [n]$, $\deg_u(Q) \leq \frac{1}{\tau^2} \max(1, n^{\frac{k}{2}-1-|Q|})$.
    \end{inparaenum}
    For convenience, let $m \geq n^{\frac{k-1}{2}} \sqrt{p} \cdot \beta$ where $\beta \coloneqq C \cdot \BipartiteBeta$ and $\gamma \coloneqq 1-2\eta \in (0,1]$ since $\eta \in [0, \frac{1}{2})$.
    
    First, we construct the $2$-XOR instance $\phi$ defined in \cref{def:2XOR-from-bipartite-kXOR}. As stated in \cref{obs:edges-in-phi}, the average degree is at least $d \coloneqq \frac{1}{4}\beta^2$, and furthermore, by \cref{lem:degree-bounds-eps-spread}, the maximum degree of $G_{u,C}^{(L)}(\phi)$ and $G_{u,C'}^{(R)}(\phi)$ for any $u\in[p]$, $C\in \cH_u^{(L)}$ and $C' \in \cH_u^{(R)}$ is bounded by $\Delta \coloneqq 1/\tau^2$.
    The algorithm then follows the steps outlined in \cref{sec:bipartite-kXOR-outline}.
    
    \parhead{Step 1.} We apply graph pruning and expander decomposition (\cref{lem:pruning-expander-decomp}) with parameter $\eps' \coloneqq \frac{1}{4}\eps$, which decomposes $\phi$ into $\phi_1,\dots,\phi_T$ such that they contain $1-\eps'$ fraction of the constraints in $\phi$, and their constraint graphs (after adding some self-loops due to expander decomposition) have minimum degree $d_{\min} \geq \frac{1}{3} \eps' d = \frac{1}{48}\eps \beta^2$ and spectral gap $\lambda \geq \Omega(\eps'^2/\log^2 m) = \Omega(\eps^2 / (k^2\log^2 n))$.

    \parhead{Step 2.} We solve the SDP relaxation for each subinstance $\phi_i$.
    Let $G$ be the constraint graph of $\phi_i$ (with at most $N \leq n^{k-1}$ vertices) and $H$ be the corrupted edges of $G$.
    We apply the relative spectral approximation result (\cref{lem:spectral-sparsification-odd}) with $\xi_1^{(1)}, \dots, \xi^{(1)}_{m/2p}$ (resp.\ $\xi_1^{(2)}, \dots, \xi^{(2)}_{m/2p}$) being $\Fits$ random variables indicating whether each $C\in \cH_u^{(L)}$ (resp.\ $C'\in \cH_u^{(R)}$) is corrupted.
    Moreover, the subgraphs $\GL{i}$ and $\GR{j}$ in \cref{lem:spectral-sparsification-odd} (which are simply subgraphs of $G_{u,C}^{(L)}(\phi)$ and $G_{u,C'}^{(R)}(\phi)$) have maximum degree $\leq \Delta = 1/\tau^2$.
    Thus, we have that with probability $1- O(N^{-2})$,
    \begin{flalign*}
        L_H \preceq \max\Paren{(1+\delta) \cdot 2\eta(1-\eta),\ \frac{1}{3}} \cdot L_G
    \end{flalign*}
    where $\delta = \sqrt{\frac{B\Delta\log N}{d_{\min} \lambda}} \leq O\Bigparen{\sqrt{\frac{k^3 \log^3 n}{\tau^2 \eps^3 \beta^2}}}$. Plugging in $\beta$ (for large enough $C$), we get that $\delta \leq \gamma^2 = 1-4\eta(1-\eta)$.
    Therefore, we have $(1+\delta) \cdot 2\eta(1-\eta) \leq (1+\gamma^2) \cdot \frac{1}{2}(1-\gamma^2) < \frac{1}{2}$, hence $L_H \prec \frac{1}{2} L_G$.
    By union bound over all $T \leq N$ subinstances, this holds for all subinstances $\phi_i$ with probability $1-\frac{1}{\poly(n)}$ over the randomness of the noise.
    
    Then, by \cref{lem:sdpuniqueness}, the SDP relaxation has a unique optimum which is the planted assignment.
    Thus, we can identify the set of corrupted edges in each $\phi_i$.
    
    \parhead{Step 3.} So far we have identified, for $\geq 1-\eps'$ fraction of all $\{(u,C,C') : u \in [p], C \in \cH_{u}^{(L)}, C' \in \cH_{u}^{(R)}\}$, the product $\xi_u(C) \cdot \xi_u(C')$, where $\xi_u(C) = -1$ if $(u,C)$ is corrupted in $\psi$, and $+1$ otherwise.
    Let $P_u \subseteq \{(C,C') : C \in \cH_{u}^{(L)}, C' \in \cH_{u}^{(R)}\}$ be such pairs for each $u\in[p]$, and let $P = \cup_{u\in[p]} P_u$.
    Note that $|P| \geq (1-\eps') \frac{m^2}{4p}$ and $P$ depends only on $\cH$ and not on the noise.
    
    We then run \cref{alg:recovery}.
    By the assumption that $\tau \leq \frac{c\gamma}{\sqrt{k\log n}}$ for a small enough $c$, we have $|\cH_u| = \frac{m}{p} \geq 2\floor{\frac{1}{2\tau^2}} \geq \frac{24k}{(1-2\eta)^2}$, which is the  condition we need in \cref{lem:recovery}.
    Thus, with probability $1- n^{-k}$, \cref{alg:recovery} outputs \begin{inparaenum}[(1)]
        \item $\cA_1 \subseteq \cH$ which only depends on $\cH$ and such that $|\cA_1| \leq 4\eps' m = \eps m$, and
        \item $\cA_2 \subseteq \cH$, the set of corrupted constraints in $\cH \setminus \cA_1$.
    \end{inparaenum}
    This completes the proof of \cref{lem:algbipartitexor}.
\end{proof}

\section*{Acknowledgements}

We would like to thank Omar Alrabiah, Sidhanth Mohanty and Jeff Xu for enlightening discussions and feedback on our paper.
We also thank anonymous reviewers for their valuable feedback.

\bibliographystyle{alpha}
\bibliography{planted-csp.bbl}

\newpage
\appendix
\section{Notions of Relative Approximation}
\label{sec:sparsification-example}

In this paper, we have encountered several notions of relative graph approximations.
Let $G$ be an $n$-vertex graph, and let $H$ be a random subgraph of $G$ by selecting each edge with a fixed probability $\eta \in (0,1)$.
We are interested in the sufficient conditions on $G$ for each of the following to hold with probability $1-o(1)$ (for some $\delta = o(1)$):
\begin{enumerate}[(1)]
    \item \textbf{Relative cut approximation}: $x^\top L_H x \leq (1+\delta) \eta \cdot x^\top L_G x$ for all $x\in \Fits^n$.
    \label{item:cut-sparsification}
    
    \item \textbf{Relative SDP approximation}: $\iprod{X, L_H} \leq (1+\delta) \eta \cdot \iprod{X, L_G}$ for all symmetric matrices $X\succeq 0$ with $\diag(X) = \Id$.
    \label{item:sdp-sparsification}

    \item \textbf{Relative spectral approximation:} $L_H \preceq (1+\delta) \eta \cdot L_G$.
    \label{item:spectral-sparsification}
\end{enumerate}
Here, we only state one-sided inequalities, as solving noisy XOR requires only an upper bound on $L_H$.
Note also that the above is in increasing order: relative spectral approximation implies relative SDP approximation, which in turn implies relative cut approximation.

Recall from \cref{lem:cutsparsification} that a lower bound on the min-cut of $G$ suffices for cut approximation to hold, while \cref{lem:spectral-sparsification} shows that lower bounds on the minimum degree and spectral gap of $G$ suffice for spectral approximation to hold.
It is natural to wonder whether a min-cut lower bound is sufficient for SDP approximation as well, since it allows us to efficiently recover the planted assignment in a noisy planted 2-XOR via solving an SDP relaxation (see \cref{lem:sdpuniqueness}).
Unfortunately, there is a counterexample.

\paragraph{Separation of cut and SDP approximation.}
The example is the same graph that separates cut and spectral approximation described in \cite{SpielmanT11}.
Let $n$ be even and $k = k(n)$.
Define $G = (V,E)$ be a graph on $N = nk$ vertices where $V = \{0,1,\dots,n-1\} \times \{1,\dots,k\}$ and $(u,i)$, $(v,j) \in V$ are connected if $v = u \pm 1 \mod{n}$.
Moreover, there is one additional edge $e^*$ between $(0,1)$ and $(n/2,1)$.
In other words, $G$ consists of $n$ clusters of vertices of size $k$, where the clusters form a ring with a complete bipartite graph between adjacent clusters, along with a special edge $e^*$ in the middle.

Clearly, the minimum cut of $G$ is $2k$, which means that cut approximation holds. Essentially, the special edge $e^*$ does not play a role here.

However, we will show that $e^*$ breaks SDP approximation.
Define vector $x_0 \in \R^V$ such that the $(u,i)$ entry is
\begin{equation*}
    x_0(u,i) = \min(u, n-u) \mcom
\end{equation*}
and vectors $x_1,\dots,x_{n-1}$ to be cyclic shifts of $x_0$: for $w\in \{0,1,\dots,n-1\}$,
\begin{equation*}
    x_w(u,i) = x_0(u-w\ (\text{mod }n),\ i) \mper
\end{equation*}
We note that $x_0$ is the vector shown in \cite{SpielmanT11} that breaks spectral approximation.
We now show that $X = \sum_{w=0}^{n-1} x_w x_w^\top$ (scaled so that $X$ has all 1s on the diagonal) breaks SDP approximation.

First, it is easy to see that the diagonal entries of $X$ are all equal due to symmetry. Thus, for some scaling $c$, $cX \succeq 0$ and $\diag(cX) = \Id$.

Observe that for $w \leq \frac{n}{2}-1$, $x_w(0,1) = w$ and $ x_w(\frac{n}{2},1) = \frac{n}{2}-w$.
For $w \geq \frac{n}{2}$, $x_w(0,1) = n-w$ and $ x_w(\frac{n}{2},1) = w-\frac{n}{2}$. Thus, as $x_w^\top L_{e^*} x_w = \bigparen{x_w(0,1) - x_w(\frac{n}{2},1)}^2$,
\begin{equation*}
    \iprod{X, L_{e^*}} = \sum_{w=0}^{n-1} x_w^\top L_{e^*} x_w = \sum_{w=0}^{\frac{n}{2}-1} \Paren{\frac{n}{2}-2w}^2 + \sum_{w=\frac{n}{2}}^{n-1} \Paren{\frac{3n}{2}-2w}^2 = \Theta(n^3) \mper
\end{equation*}
On the other hand, $x_w^\top L_{G\setminus e^*} x_w = nk^2$ for any $w$, thus $\iprod{X, L_{G\setminus e^*}} = n^2 k^2$.
This is $o(n^3)$, i.e.\ dominated by $\iprod{X, L_{e^*}}$, when $k = o(\sqrt{n})$.
Since $e^*$ is selected in $H$ with probability $\eta$, we have that with probability $\eta$,
\begin{equation*}
    \iprod{X, L_H} \geq \iprod{X, L_{e^*}} \geq (1-o(1)) \cdot \iprod{X, L_G} \mcom
\end{equation*}
which violates the desired SDP approximation.

\section{Hypergraph Decomposition}
\label{sec:hypergraph-decomp}

In this section, we describe the hypergraph decomposition algorithm used in \cref{sec:decomposition} (for the proof of \cref{mthm:algxor}).
This algorithm is nearly identical to the hypergraph decomposition step of \cite[Section 4]{GuruswamiKM22}.

\begin{mdframed}
  \begin{algorithm}
    \label{alg:decomp}\mbox{}
    \begin{description}
    \item[Given:]
       A semirandom (with noise $\eta$) $k$-XOR instance $\psi$ with constraint hypergraph $\cH$ over $n$ vertices, and a spread parameter $\tau \in (0,1)$.
    \item[Output:]
        For each $t = 2, \dots, k$, a semirandom (with noise $\eta$) planted $\tau$-spread $p^{(t)}$-bipartite $t$-XOR instance $\psi^{(t)}$ with constraint hypergraph $\{\cH^{(t)}_u\}_{u \in [p^{(t)}]}$, along with ``discarded'' hyperedges $\cH^{(1)}$.
           \item[Operation:]\mbox{}

    \begin{enumerate}
    	\item \textbf{Initialize:} $\psi^{(t)}$ to the empty instance, and $p^{(t)} = 0$ for $t = 2, \dots, k$.
	    \item \textbf{Fix violations greedily}:
			\begin{enumerate}
    			\item Find a maximal nonempty violating $\setQ$. That is, find $Q \subseteq [n]$ of size $1 \leq \abs{Q} \leq k - 1$ such that $\deg(Q) = \abs{\{C \in \cH : Q \subseteq C\}} > \frac{1}{\tau^2} \max(1, n^{\frac{k}{2}  - \abs{Q}})$, and $\deg(Q') \leq  \frac{1}{\tau^2} \max(1, n^{\frac{k}{2} - \abs{Q'}})$ for all $Q' \supsetneq Q$. 
                    \label{step:violating-Q}
                
    			\item Let $q = \abs{Q}$. Let $u = 1 + p^{(k + 1 - q)}$ be a new ``label'', and define $\cH^{(k + 1 - q)}_u$ to be an arbitrary subset of $\{C \setminus Q: C \in \cH, Q \subseteq C\}$ of size exactly $2 \cdot \floor{ \frac{1}{2\tau^2} \max(1, n^{\frac{k}{2} - q})}$.
                    \label{step:add-to-Hu}
                    
			\item Set $p^{(k + 1 - q)} \gets 1 + p^{(k + 1 - q)}$, and $\cH \gets \cH \setminus \cH^{(k+1 - q)}_{u}$.
    		\end{enumerate}
		\item If no such $\setQ$ exists, then put the remaining hyperedges in $\cH^{(1)}$. 
      \end{enumerate}
    \end{description}
  \end{algorithm}
\end{mdframed}

\begin{lemma}
    \cref{alg:decomp} has the following guarantees:
    \begin{enumerate}[(1)]
        \item The runtime is $n^{O(k)}$,
        \item The number of ``discarded'' hyperedges is $m^{(1)} \defeq \abs{\cH^{(1)}} \leq \frac{1}{k \tau^2} n^{\frac{k}{2}}$,
        \item For each $t\in \{2,\dots,k\}$ and $u\in[p^{(t)}]$, $\abs{\cH_u^{(t)}} = \frac{m^{(t)}}{p^{(t)}} = 2 \floor{\frac{1}{2\tau^2} \max(1, n^{ t - \frac{k}{2} - 1})}$,
        \item For each $t = 2, \dots, k$, the instance $\psi^{(t)}$ is $\tau$-spread.
    \end{enumerate}
\end{lemma}

\begin{proof}
The runtime of \cref{alg:decomp} is obvious. We now argue that $m^{(1)}$ is small. By construction, $\cH^{(1)}$ is the set of remaining hyperedges when the inner loop terminates, and so we must have $\deg(\{i\}) \leq \frac{1}{\tau^2} \max(1, n^{\frac{k}{2} - 1}) =  \frac{1}{\tau^2} n^{\frac{k}{2} - 1}$ for every $i \in [n]$; here, $\deg$ only counts hyperedges remaining in $\cH$. We then have $\sum_{i \in [n]} \deg(\{i\}) = k \abs{\cH^{(1)}}$, as every $C \in \cH^{(1)}$ is counted exactly $k$ times in the sum. Hence, $m^{(1)} \leq \frac{1}{k \tau^2}n ^{\frac{k}{2}}$.

 Next, for each $t \in \{2, \dots, k\}$, by construction (Step (\ref{step:add-to-Hu})) each $\cH^{(t)}_u$ has the same size, namely $2\floor{ \frac{1}{2\tau^2} \max(1, n^{ t - \frac{k}{2} - 1})}$. It then follows that $m^{(t)} := \sum_{u \in [p^{(t)}]} \abs{\cH^{(t)}_u} = p^{(t)} \cdot 2\floor{ \frac{1}{2\tau^2} \max(1, n^{ t - \frac{k}{2} - 1})}$, and so $\abs{\cH^{(t)}_u} = \frac{m^{(t)}}{p^{(t)}}$. We also note that $m^{(t)}/p^{(t)}$ is clearly even.

We now argue that for each $t$, the instance $\psi^{(t)}$ is $\tau$-spread.  From \cref{def:spread-bipartite-kXOR}, we need to prove that for each $u \in [p^{(t)}]$ and $Q \subseteq [n]$, $\deg_u(Q) \leq \frac{1}{\tau^2} \max(1, n^{\frac{k}{2} - 1 - \abs{Q}})$.
To see this, let $u \in [p^{(t)}]$, and let $Q_u$ be the set ``associated'' with the label $u$, i.e., the set picked in Step (\ref{step:violating-Q}) of \cref{alg:decomp} when the label $u$ is added in Step (\ref{step:add-to-Hu}). Note that we must have $\abs{Q_u} = k + 1 - t$. Let $\cH'$ denote the set of constraints in $\cH$ at the time when $u$ and $\cH^{(t)}_u$ is added to $\psi^{(t)}$. Namely, we have that for every $C \in \cH^{(t)}_u$, $Q_u \cup C \in \cH'$, and $Q_u, C$ are disjoint. Now, let $R \subseteq [n]$ be a nonempty set of size at most $t - 1$. First, observe that if $R \cap Q_u$ is nonempty, then we must have $\deg_u(R) = 0$ (this degree is in the hypergraph $\cH_u^{(t)}$). Indeed, this is because $C \cap Q_u = \emptyset$ for all $C \in \cH^{(t)}_u$. So, we can assume that $R \cap Q_u = \emptyset$. Next, we see that $\deg_u(R) \leq \deg_{\cH'}(Q_u \cup R)$ (where $\deg_{\cH'}$ is the degree in $\cH'$), as $Q_u \cup C \in \cH'$ for every $C \in \cH^{(t)}_u$. Because $Q_u$ was maximal whenever it was processed in our decomposition algorithm and $Q_u \subsetneq Q_u \cup R$ as $R$ is nonempty and $R \cap Q_u = \emptyset$, it follows that 
\begin{flalign*}
\deg_{\cH'}(Q_u \cup R) &\leq \frac{1}{\tau^2} \max(1, n^{\frac{k}{2} - \abs{Q_u \cup R}}) = \frac{1}{\tau^2} \max(1, n^{\frac{k}{2} - \abs{Q_u} - \abs{R} }) \\
&= \frac{1}{\tau^2} \max(1, n^{t  - \frac{k}{2} - 1 - \abs{R} }) \leq \frac{1}{\tau^2} \max(1, n^{\frac{t}{2} - 1 - \abs{R} }) \mcom
\end{flalign*}
where the last inequality follows because $t - \frac{k}{2} - 1 - \abs{R} \leq \frac{t}{2} - 1 - \abs{R}$ always holds, as $t \leq k$.

Finally, when $R= \emptyset$, we trivially have
\begin{align*}
    \deg_u(\emptyset) = \Abs{\cH_u^{(t)}} = 2\left\lfloor \frac{1}{2\tau^2} \max(1, n^{ t - \frac{k}{2} - 1}) \right\rfloor
    \leq  \frac{1}{\tau^2} \max(1, n^{ t - \frac{k}{2} - 1}) \leq \frac{1}{\tau^2} \max(1, n^{\frac{t}{2} - 1 }) \mcom
\end{align*}
where we use again that $t - \frac{k}{2} \leq \frac{t}{2}$ as $t \leq k$.
This finishes the proof.
\end{proof}

\section{\texorpdfstring{\cref{mthm:algxor}}{Theorem~\ref{mthm:algxor}} when \texorpdfstring{$k = 1$}{k = 1}}
\label{sec:1xor}

In this section, we state and prove a variant of \cref{mthm:algxor} for the degenerate case of $k = 1$. The algorithm here is straightforward, and we include it only for completeness.
\begin{lemma}[Algorithm for noisy $1$-XOR]
\label{lem:alg1xor}
Let $\eta \in (0, 1/2)$ be a constant. Let $n \in \N$ and $\eps \in (0, 1)$, and let $m \geq O(n \log n/\eps)$. There is a polynomial-time algorithm $\Alg$ that takes as input a $1$-XOR instance $\psi$ with constraint hypergraph $\cH$ and outputs two disjoint sets $\Alg_1(\cH), \Alg_2(\psi) \subseteq \cH$ with the following guarantees: \begin{inparaenum}[(1)] \item for any instance $\psi$ with $m$ constraints, $\abs{\Alg_1(\cH)} \leq \eps m$ and $\Alg_1(\cH)$ only depends on $\cH$, and \item for any $x^* \in \Fits^n$ and any $k$-uniform hypergraph $\cH$ with at least $m$ hyperedges, with high probability over $\psi \gets \psi(\cH, x^*, \eta)$, it holds that $\Alg_2(\psi) = \Err_{\psi} \cap (\cH \setminus \Alg_1(\cH))$.\end{inparaenum}
\end{lemma}
\begin{proof}
First, observe that a $1$-XOR instance is a degenerate case where $\cH$ is a multiset of $[n]$ of size $m$. Let $S \subseteq [n]$ denote the set of $i \in [n]$ where $i$ appears in $\cH$ with multiplicity $\leq c \log n$, where $c$ is a constant to be determined later. Let $\Alg_1(\cH)$ denote $\cH \cap S$, i.e., the set of elements in $\cH$ that are in $S$. We clearly have that $\abs{\Alg_1(\cH)} \leq c n \log n \leq \eps m$.

Now, let $i \notin S$. Observe that for each occurrence of $i$ in $\cH$, we have a corresponding \emph{independent} right-hand side $b \in \Fits$ where $b = x^*_i$ with probability $1 - \eta$ and $-x^*_i$ with probability $\eta$. Thus, by taking the majority, we can with high probability decode $x^*_i$ and thus determine the corrupted constraints. It thus remains to show that with probability $\geq 1 - 1/\poly(n)$, the fraction of corrupted right-hand sides for $i$ is $< \frac{1}{2}$. Indeed, by a Chernoff bound, with probability $\geq 1 - \exp(-2\delta^2 c \log n)$, it holds that the fraction of corrupted right-hand sides is at most $(\eta + \delta)$. By choosing $\delta = \frac{1}{2}(\frac{1}{2} - \eta)$ and $c$ to be a sufficiently large constant, \cref{lem:alg1xor} follows.
\end{proof}

\end{document}